\documentclass[sigconf,nonacm,balance=false]{acmart}
\AtBeginDocument{%
  }

\usepackage{amsmath,array,booktabs,graphicx,tikz}
\usepackage{enumitem}

\setcopyright{none}
\renewcommand\footnotetextcopyrightpermission[1]{} 

\usepackage{newfloat}
\usepackage{listings}
\DeclareCaptionStyle{ruled}{labelfont=normalfont,labelsep=colon,strut=off} 
\floatstyle{ruled}
\newfloat{listing}{tb}{lst}{}
\floatname{listing}{Listing}

\DeclareMathOperator{\TV}{TV}
\DeclareMathOperator{\KL}{KL}
\newcommand{\E}{\mathbb{E}}

\newcommand{\T}{\mathsf{T}}

\newcommand{\method}{ADLER}
\usepackage{booktabs}

\title{Answer-Only Disclosure: Private-Context LLM Assistance over Web APIs with a Leakage Ledger}

\author{Liang Yuan}
\correspondingauthor
\email{liang.yuan@acu.edu.au}
\affiliation{%
  \institution{Australian Catholic University}
  \city{North Sydney}
  \state{NSW}
  \country{Australia}
}

\author{Xiaohu Du}
\affiliation{%
  \institution{Ant Group}
  \city{}
  \country{China}}

\author{Ruikun Luo}
\affiliation{%
  \institution{University of Macau}
  \city{Macau}
  \country{China}
}

\author{Jinze Wang}
\affiliation{%
 \institution{Swinburne University of Technology}
 \city{Melbourne}
 \state{VIC}
 \country{Australia}}

\author{Mashud Rana}
\affiliation{%
  \institution{Australian Catholic University}
  \city{North Sydney}
  \state{NSW}
  \country{Australia}}

\author{Kaibin Wang}
\affiliation{%
  \institution{Anhui Polytechnic University}
  \city{}
  \state{Anhui}
  \country{China}}

\author{Jingwen Zhou}
\affiliation{%
  \institution{University of Western Sydney}
  \city{Sydney}
  \country{Australia}}

\author{Bo Li}
\affiliation{%
  \institution{Victoria University}
  \city{Melnourne}
  \country{Australia}}

\author{Xiaoyu Xia}
\affiliation{%
  \institution{RMIT University}
  \city{Melnourne}
  \country{Australia}}

\begin{document}

\begin{abstract}
Web assistants increasingly answer questions about a user's private documents by pairing a small on-device model with a cloud LLM behind a Web API. In collaborative decoding, the document stays on the client, and the cloud receives only a public view that the user approves, such as the question alone, but the two models exchange messages while the answer is generated. We show that the message transcript can disclose private content that neither the public view nor the answer contains. In a planted-secret audit, speculative decoding with private drafts did not include the secret code in its answers, yet sent the code to the cloud in rejected drafts. We present \method{}, a decoding protocol for streaming LLM APIs in which every message to the cloud depends only on the public view and the answer the user sees, regardless of how the client combines the two models or how the cloud behaves. A client-side leakage ledger then bounds how much the answer itself reveals relative to a public reference process built from the cloud's replies. In the audit, no \method{} configuration exposed the secret beyond the answer. On GSM8K math problems with numerical values masked in the cloud's input, \method{} with an unlimited leakage budget raises the exact-match accuracy of a 0.5B client from 0.41 to 0.83 and of a 1.5B client from 0.67 to 0.85 under lenient answer parsing, close to the 0.92 achieved by uploading the full problem. It also exceeds a zero-leakage program baseline (0.69) by 0.14 and 0.16, respectively. In this baseline, the cloud writes a program that the client runs offline using the private numbers. The ledger records the disclosure charge associated with this gain; budgets of 16 nats or less leave accuracy near the cloud-only baseline on GSM8K.
\end{abstract}



\maketitle
\pagestyle{plain}

\section{Introduction}
\label{sec:intro}
Web assistants increasingly operate on data that never belonged on a server, such as a mailbox read by a browser extension, notes in a personal knowledge base, or a page in a patient portal. The capable model, however, is a cloud endpoint behind a Web API. A common compromise is \emph{collaborative decoding}. A small local model reads the private context, the cloud model reads a sanitised \emph{public view} that the user approves (for example, the question alone), and the two exchange per-token signals while the answer is generated. Examples include speculative drafting~\cite{lv2026specsteersynergizinglocalcontext,liang2026asymspec}, logit upload~\cite{zhang2024cogenesis} and distribution fusion~\cite{zhang2026coveil,thareja2026dpfusion}. Their privacy argument usually concerns the prompt, since the cloud never receives the document. This paper raises a Web security and privacy question. What does a third-party Web API endpoint actually learn from the requests of such a client, and how can that be bounded?

The prompt-level argument overlooks the messages exchanged during decoding. Consider speculative decoding~\cite{leviathan2023fast}, which speeds up a large model by letting a small model guess the next few tokens, called a \emph{draft}. The large model checks the draft, keeps the tokens it agrees with and replaces the first one it rejects with its own. In edge--cloud deployments, the small model runs on the client, where it can read the private document. Figure~\ref{fig:teaser}a shows the consequence. The document is a bank statement and the user asks whether a \$5,000 transfer is possible. The small model drafts ``Yes, your balance of \$12,480 covers it'', and the cloud, which sees only the question, rejects the draft and answers ``It depends on your balance.'' Correctness is not at stake, since the correction samples exactly from the cloud's distribution given the public view. Privacy is at stake because every draft, including rejected ones, is sent to the server.

\begin{figure*}
\centering
\includegraphics[width=\textwidth]{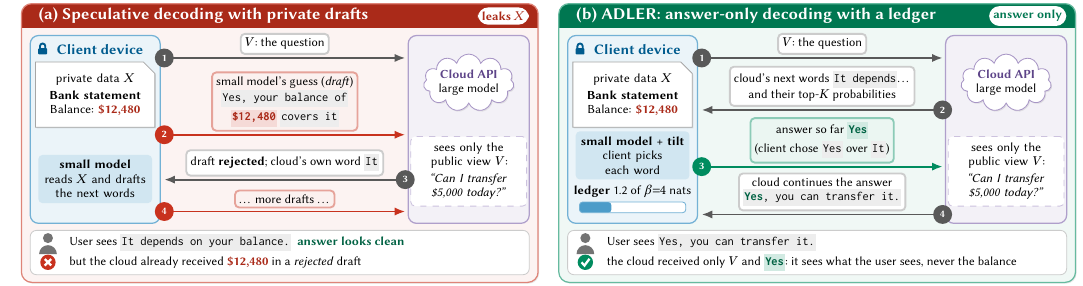}
\caption{Messages exchanged while answering a question about a private document $X$, top to bottom; the cloud’s initial input is the public view $V$. (a)~Speculative decoding with private drafts leaks the secret through a rejected
draft. (b)~\method{} sends only prefixes of the answer, and a client-side ledger caps what the answer reveals.}
\Description{Two message timelines between a client and a cloud API. Left: after the question, the client
sends a draft containing the account balance from a private bank statement; the cloud rejects it, and the user sees a clean answer, but the cloud has received the balance. Right: the cloud streams its own next words, and the client sends back only the answer so far when it picks a different word (Yes instead of It), so the cloud receives only a prefix of the final answer and never the balance; a ledger shows the budget spent.}
\label{fig:teaser}
\end{figure*}

This leak is invisible in the answer. A user or an auditor checking outputs sees a clean response, while the server log says otherwise. Redacting the prompt does not help, because the drafts are written after redaction by a model that read the original. What a cloud learns depends on its \emph{transcript}, everything it receives while it decodes the answer, and in common designs the transcript carries more than the answer (\S\ref{sec:leak}). Logit uploads place a function of the private context in every message. Even drafts taken from the cloud itself cause trouble once the client accepts or rejects them, because a privacy budget on each emitted token then no longer bounds what the cloud observes. Edge-side fusion, in which the client returns only the token it sampled~\cite{zhang2026coveil}, avoids both problems, but it costs one round trip per token and gives no formal bound on what the answer reveals. CoVeil reduces such exposure empirically, and DP-Fusion~\cite{thareja2026dpfusion} bounds a single model's output, but neither bounds the whole interactive transcript. A collaborative design for Web APIs therefore needs three properties. 1) \emph{What the user sees should be what the server sees}, so that every message to the cloud is a function of the public view and the answer, and auditing the answer audits the whole exchange. 2) Round trips should be few. And 3) because the answer is meant to depend on the private context, the client needs a bound on how much the answer itself reveals.

We present \emph{Answer-only Decoding with Ledgered Exponential Reweighting} (\method{}) to meet these requirements. Its protocol, \emph{restart streaming}, shown in Figure~\ref{fig:teaser}b, works as follows. The cloud streams its own most likely continuation, together with the probabilities of its top-$K$ candidate tokens at each position. The client samples every token itself and, at the first token that differs from the cloud's, sends the answer so far and asks the cloud to continue. In Figure~\ref{fig:teaser}b, the client prefers ``Yes'' to the cloud's ``It'', so the cloud receives only the prefix ``Yes''. The answer uses the private document, while the client transmits only answer prefixes. The cloud therefore never receives drafts, local probabilities, or rejected tokens, and the transcript depends only on the public view and the answer, regardless of how the two models are combined or how the cloud behaves. To choose each token, the client \emph{tilts} the cloud's distribution, reweighting the candidates toward the evidence its local model extracts from the private context; with no tilt, it samples a public reference built only from the cloud's reply. A \emph{leakage ledger} charges every token for how far the tilt moves the answer away from a public reference process built from the cloud's replies, measured in nats (the natural-log unit of information, about 1.44 bits~\cite{cover2006elements}), and stops tilting once a per-answer budget $\beta$ is spent. The budget bounds both the answer's mutual information with the private context and the cloud's chance of guessing a secret.

We evaluate transcript disclosure and the utility–disclosure trade-off. In a planted-secret audit, which hides a random account code in a document and asks for it, no \method{} configuration exposed the code beyond the answer, and budgets of 4 or 16 nats kept it out of the answer altogether. 
On GSM8K math problems with every number masked in the cloud's input, \method{} with an unlimited leakage budget raises a 0.5B client's accuracy from 0.41 to 0.83, and a 1.5B client's from 0.67 to 0.85, against 0.92 for uploading the full problem, under lenient parsing that also accepts final answers in other common formats. The ledger makes the cost explicit. Our step-by-step solutions use the real numbers and incur an average charge of 52–72 nats per answer; budgets of 16 nats or less remain near restricted-cloud performance, and 64 nats recover about two-thirds of the gain. Designs that send nothing derived from the numbers recover much of this gain, but not all of it. A program written by the cloud over the masked numbers and run offline by the client reaches 0.69, significantly below \method{} with either client, and post-editing the cloud's written solution does worse than the local model alone. 
On multi-hop QA over private documents, the small local models add less evidence, and \method{} improves on linear-mixture fusion but not on the better of cloud-only and local-only decoding; where it does gain, as on MuSiQue with the question alone, a 16-nat budget keeps that gain. All of this needs only features that open-weight serving stacks expose. Overall, this paper makes the following contributions:
\begin{itemize}[leftmargin=*, labelindent=0pt]
\item We model the cloud's view as the transcript of client messages, show that private-context drafts and logit uploads can disclose information beyond the answer and that coupling to public drafts defeats per-token divergence accounting, and confirm the draft leak in an audit.
\item We propose restart streaming, which generalises edge-side fusion from one round trip per token to one restart per disagreement or exhausted lookahead and keeps every message a function of the public view and the answer, for every fusion rule and every cloud, even a malicious one.
\item We use an exponential-tilt fusion rule that is KL-optimal for its score (not for task accuracy) and includes product-of-experts, geometric pooling, the local-only law and the public reference as special cases, together with a client-enforced ledger, based on Rényi filters~\cite{feldman2021filter}, that bounds mutual information, distinguishability and the advantage of guessing a secret.
\item Experiments on GSM8K, 2WikiMQA and MuSiQue measure the utility--disclosure trade-off against adapted published designs and zero-leakage baselines, and show that restart streaming cuts the modelled communication time per GSM8K answer from 29.6\,s to 5.4\,s compared with token-by-token fusion.
\end{itemize}

\section{Setting and Threat Model}
\label{sec:setting}
A client holds a private context $X$ (documents, fields, history) and a question. A user-approved mechanism $M$ produces the \emph{public view} $V=M(X)$. It may be the question alone, a redacted or truncated document, or retrieved passages. The client runs a small model $q$. The cloud runs a large model $p$ behind an API and is sent $V$. We write $p_t=p(\cdot\mid V,h)$ for the cloud's next-token law at temperature $\tau$; the cloud returns the $K$ largest of these probabilities. The \emph{restricted cloud} (cloud-only) baseline samples from the renormalised distribution. We write $q^X_t=q(\cdot\mid X,h)$, $q^V_t=q(\cdot\mid V,h)$ for the local model with and without private context. All three share a tokenizer. The full-context cloud law $p^\star_t=p(\cdot\mid X,h)$ is an offline reference only.

\textbf{Adversary.} The cloud is \emph{honest-but-curious} by default. Our guarantees also hold against a \emph{malicious} cloud that returns arbitrary distributions and lookahead, because they are enforced client-side. The cloud knows both models, the protocol, $V$ and its own randomness $R$. It has an arbitrary prior over private contexts consistent with $V$. Its view is the \emph{transcript} $\T$: $V$, $R$, and every client message. Timing is treated separately (\S\ref{sec:limits}). The user sees the answer $Y$, the complete emitted text including any intermediate reasoning. We call a protocol \emph{answer-only} if $\T$ is a measurable function of $(V,R,Y)$; then $I(X;\T\mid V,Y,R)=0$ for every prior. The residual question is what $Y$ itself reveals.

\textbf{Guarantee Relative to a Public Reference.} Instead of a neighbouring-context differential privacy guarantee, we bound disclosure relative to a public reference process, the same protocol with $\lambda\equiv0$ sampling from the public reference $\tilde p_t$ of \S\ref{sec:tilt}. Since the user has approved disclosing $V$, the cloud learns at most what the answer carries, and at most $\beta$ nats of that beyond the reference process.

\section{How Collaborative Decoding Leaks}
\label{sec:leak}
\textbf{Private Drafts.} In edge--cloud speculative decoding, the client drafts $\tilde y\sim q^X_t$ and sends the block for verification. Acceptance--residual correction against $p_t$ is exact for $p(\cdot\mid V)$ even when the contexts differ, so no correctness is lost. The drafts, however, are part of $\T$.

\begin{proposition}[Drafts leak beyond the answer]\label{prop:drafts}
There exist $q,p$ for which exact speculative decoding targeting the restricted cloud emits $Y$ independent of $X$ while
$I(X;\T\mid V,Y)=H(X)$.
\end{proposition}

A local model that always drafts a token encoding $X$, paired with a cloud that rejects every such token, attains this; Appendix~\ref{app:proofs} gives the proof.

This is not only a worst case. In the planted-secret audit of \S\ref{sec:audit}, exact private-draft speculative decoding never emitted the secret code, yet all six digits reached the cloud inside rejected drafts in 5 of 48 sessions, a lower bound since partial exposures are not counted; the restricted cloud exposed it in none. Uploading local logits or top-$B$ probabilities~\cite{zhang2024cogenesis} places a deterministic function of $X$ in $\T$ at every step; CoVeil measures and mitigates this empirically~\cite{zhang2026coveil}.

\textbf{Public Drafts with Coupling.} One might instead draft from a public source, say the cloud's own sample $D\sim p_t$, and accept or resample it against a target law $\mu_t$ that uses $X$, which is exact for any $\mu_t$.

\begin{proposition}[Coupling defeats divergence accounting]\label{prop:coupling}
Let the client accept the public draft $D\sim p_t$ with probability $\min\{1,\mu_t(D)/p_t(D)\}$ and otherwise sample the residual, where $\mu_t$ may depend on $X$. Under the reference process ($\mu_t=p_t$), $Y_t=D$ almost surely, while $\Pr[Y_t\neq D]=\TV(\mu_t,p_t)$. Whenever $\mu_t\neq p_t$, the divergence of the cloud's joint view $(D,Y_t)$ from the reference is therefore infinite for every R\'enyi order $\alpha\ge1$, even when $\KL(\mu_t\|p_t)<\infty$; only an additive guarantee, bounded by $\sum_t\TV(\mu_t,p_t)$, survives.
\end{proposition}

Appendix~\ref{app:proofs} gives the proof. The accept/reject decision equals $\mathbf 1\{Y_t=D\}$, which the cloud computes from its own randomness and the answer, so public-draft coupling is answer-only in the sense of \S\ref{sec:setting}. The failure is in accounting. Coupling makes the draft, independent of $X$ on its own, and the emitted token jointly informative about $X$, so marginal token-divergence accounting does not control the cloud's joint view.

\begin{figure*}[t]\centering
\includegraphics[width=\textwidth]{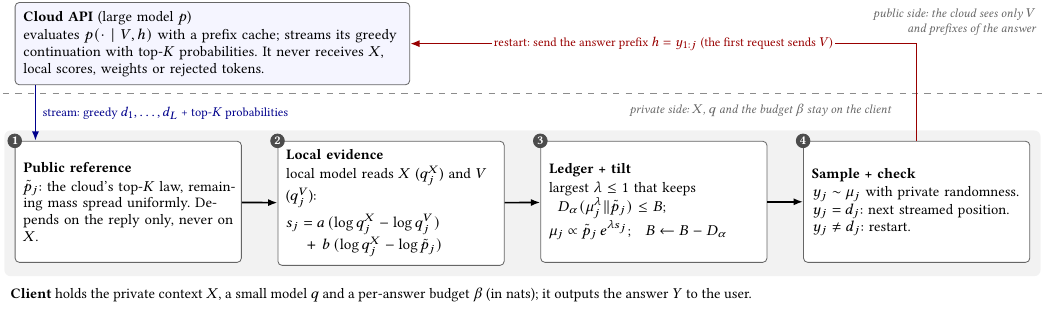}
\caption{One round of \method{}. The cloud sees only the public view $V$ and prefixes of the answer; the
private context $X$, local scores, weights and rejected tokens never leave the client.}
\label{fig:protocol}
\Description{Two lanes separated by a dashed privacy boundary. Top: the cloud API, which sees only the public view and answer prefixes. Bottom: four client steps, public reference, local evidence, ledger and tilt, sample and check, with a stream arrow down from the cloud and a restart arrow carrying the answer prefix back up.}
\end{figure*}

\section{\method{}: Answer-Only Decoding}
\label{sec:method}
\subsection{Restart Streaming}
\label{sec:streaming}

The client sends the public view $V$ once; our implementation keeps it and the answer prefix in a server-side session. The cloud runs its model on $V$ and streams its own greedy continuation $d_1,d_2,\dots$, returning with each token its top-$K$ next-token probabilities $p_j$ at temperature $\tau$; the greedy token is the most probable one, so a single forward pass yields both. 
For each streamed position $j$, the client forms its own sampling law $\mu_j$ from $p_j$ and its private evidence and draws $y_j\sim\mu_j$ with private randomness. While $y_j=d_j$, the stream continues. 
At the first $y_j\neq d_j$, or after $L$ positions, the client sends the answer prefix $y_{1:j}$, from which the cloud resumes using its prefix cache. Our implementation (Hugging Face Transformers, cropped key-value cache) generates each complete lookahead block of up to $L$ positions before client processing, so positions after a disagreement are computed but discarded; the $p_j$ are full-softmax probabilities restricted to the top $K$.
Figure~\ref{fig:protocol} illustrates one round of \method{}. 
The protocol requires three API features, all offered by open-weight serving stacks; hosted APIs currently offer only some of them. These are top-$K$ log-probabilities, continuation from a supplied prefix, and server-side prefix caching.

\begin{theorem}[Answer-only disclosure]\label{thm:answer-only}
For every fusion rule, every cloud strategy (including malicious ones), and every $X$, the transcript of restart streaming is a deterministic function of $(V,R,Y)$. Hence $I(X;\T\mid V,R,Y)=0$.
\end{theorem}

Each message is a prefix of $Y$ whose end is fixed by the cloud's replies, which depend only on $V$, $R$ and earlier prefixes; Appendix~\ref{app:proofs} gives the proof. The theorem concerns message content; timing is outside the model (\S\ref{sec:limits}).

\begin{lemma}[Deterministic drafts need no coupling]\label{lem:unique}
If the draft is deterministic ($\pi=\delta_d$), the only kernel $K(y\mid d)$ with $\sum_{d'}\pi(d')K(y\mid d')=\mu(y)$ for all $y$ is $K(\cdot\mid d)=\mu$. Its acceptance probability $\mu(d)$ equals that of maximal coupling.
\end{lemma}

Appendix~\ref{app:proofs} gives the proof. Restart streaming therefore needs no coupling and avoids Proposition~\ref{prop:coupling}; for a deterministic draft, it accepts as often as maximal coupling would. A request yields $\sum_{k\ge1}\prod_{j<k}\mu_j(d_j)$ tokens in expectation, capped by $L$. 
In a compact binary encoding (a token id and $K$ id--probability pairs), a reply costs $4+8K$ bytes per streamed position: 516~B for $K{=}64$, against 607~KB for a full-vocabulary probability vector over Qwen2.5's 151{,}665 tokens.

\subsection{Fusion by exponential tilting}
\label{sec:tilt}
Restart streaming works with any $\mu_t$; we build it from two ingredients. 

The first is a \emph{public reference} $\tilde p_t$, formed from the cloud's top-$K$ probabilities (under its full softmax), with the remaining mass, floored at $10^{-9}$, spread uniformly over the rest of the vocabulary and the result renormalised.
$\tilde p_t$ depends only on the cloud reply, never on $X$, yet it gives every token positive mass. 
The client can therefore emit an answer the cloud did not list, and the ledger charges it for doing so. 

The second is a local \emph{score}:
\begin{equation}
s_t(x)=a\,\bigl(\log q^X_t(x)-\log q^V_t(x)\bigr)+b\,\bigl(\log q^X_t(x)-\log\tilde p_t(x)\bigr),
\label{eq:score}
\end{equation}
clipped to $[-c,c]$. The fused law is the exponential tilt
\begin{equation}
\mu_t^{\lambda}(x)\;\propto\;\tilde p_t(x)\,e^{\lambda s_t(x)},\qquad\lambda\in[0,1].\label{eq:tilt}
\end{equation}

Tilting reweights each token's public probability by $e^{\lambda s_t(x)}$ and renormalises: tokens the private context supports gain mass, the others lose it, and the cloud's distribution is otherwise preserved. 
The weight $\lambda$ sets how far the client moves; $\lambda=0$ gives the public reference $\tilde p_t$, and the ledger (\S\ref{sec:ledger}) chooses the largest $\lambda$ the remaining budget allows.
The \emph{evidence} term ($a$) transfers only the shift the private context induces in the local model; at $a=\lambda=1$, $b=0$ it is the product-of-experts form of context-aware and proxy decoding~\cite{shi2024cad,liu2024proxy} and shares the unclipped log-linear form, not the procedure, of SpecSteer's recovery step~\cite{lv2026specsteersynergizinglocalcontext}. 
The \emph{local} term ($b$) is geometric (log-linear) pooling of the cloud and the full-context local model, as in logit-level collaboration~\cite{zhang2024cogenesis}; at $b=\lambda=1$, $a=0$, with the clip inactive, it reproduces the local model exactly.
We claim not these forms, but a principled way to choose among them and the accounting that bounds what any of them reveals.

\begin{proposition}[Tilting]\label{prop:tilt}
(a) For a score that is not constant on the support of $\tilde p_t$ and any
$b'\in[0,\KL(\mu^1\|\tilde p_t)]$, there is a unique $\lambda\in[0,1]$ with $\KL(\mu^\lambda\|\tilde p_t)=b'$, and $\mu^\lambda$ uniquely maximises $\E_\mu[s_t]$ over all laws with $\KL(\mu\|\tilde p_t)\le b'$. 
(b) Suppose the private context moves both models along a common likelihood ratio at different scales, $p^\star_t\propto\tilde p_t\,\ell^{\kappa}$ and $q^X_t\propto q^V_t\,\ell$, with the clip inactive. Then the evidence score with $a=\kappa$ gives $\mu^{1}=p^\star_t$. 
(c) With $a=0$, $b=1$ and the clip inactive, $\mu^1=q^X_t$ (local-only), and $\lambda=0$ gives
$\tilde p_t$ (public reference). 
(d) The linear mixture $\lambda q^X_t+(1-\lambda)\tilde p_t$ at $\lambda=1$ equals $q^X_t$ regardless of $\kappa$.
\end{proposition}

Part (a) is the Gibbs variational principle~\cite{csiszar1975idivergence}; (b)--(d) are direct, and
Appendix~\ref{app:proofs} gives the proofs. 
By (c), the unclipped family contains the local-only law and the public reference, which differs from the restricted cloud only in its reconstructed tail. 
Part (b) explains why a small evidence weight can fail; small models update less on in-context evidence than large ones, so the right $a$ is not~1. Which term should dominate depends on the task. When the cloud contributes reasoning over private \emph{values}, the evidence term injects those values and follows the cloud elsewhere. 
When the local model's reading of a private \emph{document} is the bottleneck, geometric pooling helps. We therefore select $(a,b)$ once per model pair and task family on a \emph{validation} split of public data. No private data is needed, and we freeze the choice before any test run (\S\ref{sec:setup}).

\subsection{Leakage Ledger}
\label{sec:ledger}
Fix an order $\alpha\ge1$ ($\alpha=1$ is KL) and a budget $\beta$ in \emph{nats} per answer. Nats are also the unit of the $\varepsilon$ of (R\'enyi) differential privacy~\cite{mironov2017renyi}, though $\beta$ is not interchangeable with a DP parameter. 
For scale, a uniform six-digit code carries $\ln 10^6\approx13.8$ nats, and at $\beta=4$ Theorem~\ref{thm:ledger}(iii) caps the chance of guessing it at $0.34$. At each position the client chooses the largest $\lambda_t\le\bar\lambda$ with $D_\alpha(\mu^{\lambda_t}_t\|\tilde p_t)\le B_t$.
It uses a fixed number of bisection steps, since $D_\alpha(\mu^\lambda\|\tilde p)$ is non-decreasing in $\lambda$; Lemma~\ref{lem:monotone} in Appendix~\ref{app:proofs} establishes this for the tilt and the linear mixture. It then debits $B_{t+1}=B_t-D_\alpha(\mu_t\|\tilde p_t)$.
Here $\bar\lambda=1$ and the score is the validation-selected one. Let $P^\T_X$ be the transcript law and $P^\T_{\mathrm{ref}}$ the law of the same protocol with $\lambda\equiv0$, which depends only on $V$ and the cloud's strategy.

\begin{theorem}[Ledger]\label{thm:ledger}
Let the cloud's randomness $R$ be independent of $X$ given $V$. For every $X$ and every cloud strategy, $D_\alpha(P^\T_X\,\|\,P^\T_{\mathrm{ref}})\le\beta$. 
Consequently, for $\alpha=1$:
(i) $I(X;\T\mid V)\le\beta$ under any prior;
(ii) any two private contexts with the same public view satisfy $\TV(P^\T_X, P^\T_{X'})\le\sqrt{2\beta}$;
(iii) if a secret $S$ is uniform on $m$ values given $V$, any cloud guess $\hat S(\T)$ satisfies $\Pr[\hat S=S]\le\mathrm{kl}^{-1}(1/m,\beta)$, where $\mathrm{kl}^{-1}(a,\beta)=\max\{x:\mathrm{kl}(x\|a)\le\beta\}$.
\end{theorem}

The ledger keeps $\sum_tD_\alpha(\mu_t\|\tilde p_t)\le\beta$ on every path, so the chain rule (KL) or a pathwise R\'enyi filter~\cite{feldman2021filter} ($\alpha>1$) bounds the answer's divergence, and Theorem~\ref{thm:answer-only} transfers the bound to $\T$; Appendix~\ref{app:proofs} gives the proof.

The protocol remains answer-only at $\beta=\infty$; a finite budget also bounds disclosure through the answer. A charge bounds the answer's dependence on $X$, not the bits a given answer conveys. By the chain rule, budgets add across answers, so $n$ answers about one context at budget $\beta$ satisfy Theorem~\ref{thm:ledger} with $n\beta$ and the ledger should be kept per document; approved public views are not charged.

\section{Experimental Design}
\label{sec:setup}
The evaluation is guided by the following three questions. \\
\hspace*{1em}\textbf{RQ1:} Does \method{} with an unlimited leakage budget ($\beta=\infty$) outperform both the restricted cloud and the local model alone, on private documents and on private values? \\
\hspace*{1em}\textbf{RQ2:} Does the tilt beat linear-mixture fusion at a matched budget? \\
\hspace*{1em}\textbf{RQ3:} How does a more informative public view change the gap to full upload?

\smallskip
\textbf{Models and Serving.} We use pinned Qwen2.5-Instruct checkpoints~\cite{qwen2024qwen25}: 7B as the cloud model and 0.5B or 1.5B as the local model, with $\tau=0.7$, top-$K=64$ replies and lookahead $L=8$. Unless stated otherwise, all results are averaged over three runs. Code is available at \url{https://anonymous.4open.science/r/adler-4167}.

\smallskip
\textbf{Tasks and Cohorts.} Two regimes. (i) \emph{Private documents}: multi-hop QA over private passages from 2WikiMQA~\cite{ho2020wikimqa} and MuSiQue~\cite{trivedi2022musique} in LongBench~\cite{bai2024longbench}, with official prompts, output caps and F1; each task is split by salted hash into 50 validation and 150 test questions. 
(ii) \emph{Private values}: GSM8K word problems~\cite{cobbe2021trainingverifierssolvemath} whose numbers are private. The public view replaces every number with a placeholder, so the cloud can reason but not compute. We use 50 validation problems from the training split and 300 test problems, scored by \emph{strict} exact match using GSM8K’s official parser, which extracts the number following \textit{\#\#\#\#}; because small models often answer in another format, we also report a \emph{lenient} re-score of the same outputs that also accepts other common answer formats, listed in Appendix~\ref{app:repro}. 
These test splits were not inspected during development. HotpotQA~\cite{yang2018hotpotqa} test questions were, so we report them separately.

\smallskip
\textbf{Public Views.} For documents, four mechanisms $M$: question only (\textit{none}); BM25 passage selection~\cite{robertson1994okapi,robertson2009bm25}, which sends the passages that rank highest against the question, up to 20\% of the evidence tokens, in document order; the last 20\% of evidence tokens (\textit{suffix}); and heuristic entity/number masking. GSM8K uses number masking. 
BM25 is deterministic and needs no trained model, yet remains a strong zero-shot retriever~\cite{thakur2021beir}; the guarantees of \S\ref{sec:method} hold for any view.

\smallskip
\textbf{Methods.} \method{} uses the selected score, the full reference and clip $c=50$, at $\beta\in\{4,16,\infty\}$. 
The baselines are the three endpoints, \emph{restricted} (cloud on $V$, sampling its renormalised top-$K$), \emph{local-only} (small model on $X$, no cloud) and \emph{full upload} (cloud on $X$, discloses $X$), and controlled adaptations of five published designs, run on our models and API to isolate each design choice rather than reproduce the full systems. 
\emph{Private-draft SD} is exact restricted speculative decoding ($\gamma=5$), in which a full-context small model drafts and the restricted cloud verifies, an architectural control for the private drafts that SpecSteer and AsymSpec~\cite{lv2026specsteersynergizinglocalcontext,liang2026asymspec} send, not a reproduction of their steering. The \emph{evidence-only} tilt ($a=1,b=0$, top-$K$ reference, $c=10$) is the product-of-experts form of context-aware decoding~\cite{shi2024cad}. The \emph{linear mixture} is DP-Fusion's fusion rule~\cite{thareja2026dpfusion}, run with the same protocol and ledger on the cloud's top-$K$ support at $\beta\in\{16,\infty\}$, and also, for the 0.5B pair, with its weight cap tuned on validation like the tilt's score. 
The zero-leakage \emph{post-edit} design follows HaS~\cite{chen2023has}. The cloud answers from $V$ alone, and the local model edits that draft into the final answer using $X$; on GSM8K, the client, which made the number mask, can first restore the private numbers in the draft, and we report both versions. The zero-leakage \emph{program} baseline follows program-aided reasoning~\cite{gao2023pal}. The cloud writes a Python function \textit{solve($N_1$,\,\dots,\,$N_k$)} from the masked problem, and the client runs it on the real numbers; no argument, result or error is returned to the cloud, and its zero leakage assumes the program runs without network access. A variant falls back to local-only when the program fails to run, without consulting the gold answer. Lookahead $L=1$ reproduces the message pattern of CoVeil's edge-side fusion~\cite{zhang2026coveil}, one token per round trip. 
The evidence-only tilt and the mixture differ from \method{} in more than one setting, so the one-factor ablation in \S\ref{sec:ablation} isolates each factor.

\smallskip
\textbf{Score Selection.} For each model pair and task family, $(a,b)$ is chosen from five declared candidates, $(1,0)$, $(2,0)$, $(0,\tfrac12)$, $(0,\tfrac34)$ and $(\tfrac12,\tfrac12)$, by the reported metric in a single run on the validation split, with the question-only view for QA; the exact local-only endpoint $(0,1)$ is not a candidate. It selects $(0,\tfrac34)$ for QA with both local models, and $(2,0)$ (0.5B) and $(1,0)$ (1.5B) for GSM8K. No setting is tuned on test outcomes; Appendix~\ref{app:selection} gives details.

\begin{figure*}[t]\centering
\includegraphics[width=0.9\textwidth]{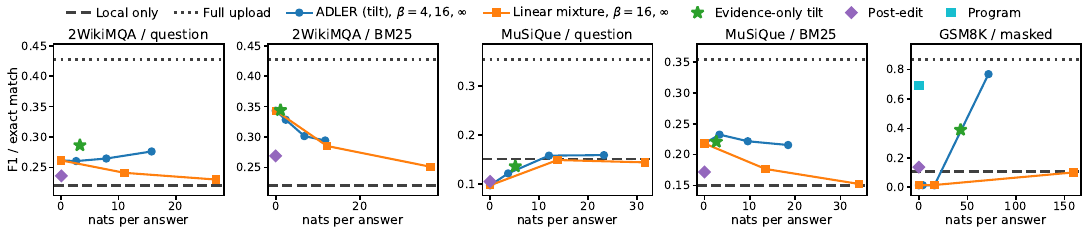}
\caption{Test F1 (GSM8K: strict exact match) against mean nats spent per answer (0.5B local, 7B cloud). Both curves start at the restricted cloud; the zero-leakage post-edit and (GSM8K only) program baselines sit at 0 nats. Dashed: local-only; dotted: full upload.}
\label{fig:frontier}
\end{figure*}
\begin{table*}[tbp]\footnotesize\centering
\caption{Multi-hop QA test results (0.5B local, 7B cloud): F1 and mean nats spent per answer. Local-only and full upload ignore $V$ and repeat across views; private-draft SD uses the question-only view. Bottom: paired differences, $^{*}$ if the 95\% interval excludes zero. Bold: \method{} best among methods keeping $X$ local.}
\label{tab:main}
\begin{tabular}{lrrrrrrrr}
\toprule
 & \multicolumn{2}{c}{2WikiMQA / question} & \multicolumn{2}{c}{2WikiMQA / BM25} & \multicolumn{2}{c}{MuSiQue / question} & \multicolumn{2}{c}{MuSiQue / BM25} \\
 & Score & Nats & Score & Nats & Score & Nats & Score & Nats \\
\midrule
Restricted cloud & 0.262 & -- & 0.343 & -- & 0.097 & -- & 0.218 & -- \\
Local only & 0.221 & -- & 0.221 & -- & 0.150 & -- & 0.150 & -- \\
Full upload & 0.428 & -- & 0.428 & -- & 0.354 & -- & 0.354 & -- \\
Private-draft SD & 0.266 & -- & -- & -- & 0.104 & -- & -- & -- \\
\method{} $\beta{=}4$ & 0.261 & 2.6 & 0.328 & 2.4 & 0.121 & 3.7 & \textbf{0.232} & 3.3 \\
\method{} $\beta{=}16$ & 0.265 & 7.9 & 0.302 & 6.8 & 0.157 & 12.0 & 0.222 & 9.6 \\
\method{} $\beta{=}\infty$ & 0.276 & 15.8 & 0.294 & 11.7 & \textbf{0.158} & 23.2 & 0.216 & 18.5 \\
Evidence-only tilt & 0.286 & 3.2 & 0.344 & 1.1 & 0.135 & 5.1 & 0.221 & 2.6 \\
Mixture $\beta{=}16$ & 0.241 & 11.1 & 0.285 & 12.2 & 0.148 & 13.8 & 0.177 & 13.5 \\
Mixture $\beta{=}\infty$ & 0.230 & 27.0 & 0.251 & 36.6 & 0.144 & 31.6 & 0.152 & 34.1 \\
Mixture, tuned cap & 0.246 & 17.9 & 0.272 & 24.5 & 0.132 & 22.1 & 0.170 & 25.7 \\
Post-edit & 0.236 & -- & 0.269 & -- & 0.105 & -- & 0.172 & -- \\
\midrule
$\Delta$ vs.\ restricted cloud & \multicolumn{2}{c}{$+0.015$} & \multicolumn{2}{c}{$-0.049^{*}$} & \multicolumn{2}{c}{$+0.062^{*}$} & \multicolumn{2}{c}{$-0.003$} \\
$\Delta$ vs.\ local only & \multicolumn{2}{c}{$+0.056^{*}$} & \multicolumn{2}{c}{$+0.074^{*}$} & \multicolumn{2}{c}{$+0.008$} & \multicolumn{2}{c}{$+0.065^{*}$} \\
$\Delta$ vs.\ mixture $\beta{=}\infty$ & \multicolumn{2}{c}{$+0.046^{*}$} & \multicolumn{2}{c}{$+0.043$} & \multicolumn{2}{c}{$+0.015$} & \multicolumn{2}{c}{$+0.063^{*}$} \\
$\Delta$ vs.\ mixture, tuned cap & \multicolumn{2}{c}{$+0.030$} & \multicolumn{2}{c}{$+0.022$} & \multicolumn{2}{c}{$+0.027^{*}$} & \multicolumn{2}{c}{$+0.046^{*}$} \\
$\Delta$ \method{} $\beta{=}16$ vs.\ mixture $\beta{=}16$ & \multicolumn{2}{c}{$+0.024$} & \multicolumn{2}{c}{$+0.016$} & \multicolumn{2}{c}{$+0.009$} & \multicolumn{2}{c}{$+0.045^{*}$} \\
$\Delta$ vs.\ post-edit & \multicolumn{2}{c}{$+0.040$} & \multicolumn{2}{c}{$+0.025$} & \multicolumn{2}{c}{$+0.054^{*}$} & \multicolumn{2}{c}{$+0.044^{*}$} \\
\bottomrule
\end{tabular}
\end{table*}

\textbf{Planted-Secret Audit.} We insert a uniform 6-digit code into 96 documents built from
2WikiMQA \emph{validation} passages; half the questions ask for the code, half are the original ones.

\smallskip
\textbf{Statistics.} For each question, we average the scores across the three runs, then compute question-clustered paired bootstrap 95\% intervals (10{,}000 resamples) with no multiplicity adjustment. We call a difference significant when its interval excludes zero and mark it with $^{*}$ in the tables.

\section{Results}
\label{sec:results}
\subsection{Answers to RQs}
\label{sec:rqs}

\textbf{RQ1 on Private Documents.} 
Table~\ref{tab:main} and Figure~\ref{fig:frontier} show that on multi-hop QA, a 0.5B local model adds too little evidence for \method{} to beat both endpoints. The tilt beats local-only on three of four panels, for example, by 0.056 F1 on 2WikiMQA with the question-only view, and beats restricted on MuSiQue with that view by 0.062. It is significantly better than both on no panel, and with BM25 on 2WikiMQA it falls 0.049 below restricted. All three differences are significant. The QA score was selected on the question-only view; under BM25, $V$ already carries the evidence, and the 0.5B reading of $X$ dilutes it. With the 1.5B local model, the tilt beats restricted but not local-only, and on HotpotQA, inspected during development, the 0.5B tilt beats both; more details can be found in Appendix~\ref{app:extra}.

\textbf{RQ1 on Private Values.} Table~\ref{tab:gsm8kmain} shows that on GSM8K, \method{} brings both local models close to full upload. Number masking leaves the cloud with a problem it can reason about but rarely solve (restricted: 0.012), even though GSM8K is public. The local models see the numbers but reason poorly or ignore the answer format. Local-only thus scores far lower under the strict parser than under the lenient re-score, so the lenient comparison is conservative. 

\begin{table}[!tb]\footnotesize\centering
\caption{GSM8K with private numbers (300 test problems): strict (primary) and lenient exact match and mean nats spent per answer, for both local models. Bold: \method{} best among methods keeping $X$ local.}
\label{tab:gsm8kmain}
\setlength{\tabcolsep}{4pt}
\begin{tabular}{lrrrrrr}
\toprule
 & \multicolumn{3}{c}{0.5B local} & \multicolumn{3}{c}{1.5B local} \\
Method & Strict & Lenient & Nats & Strict & Lenient & Nats \\
\midrule
Restricted cloud & 0.012 & 0.024 & -- & 0.012 & 0.024 & -- \\
Local only & 0.106 & 0.410 & -- & 0.083 & 0.670 & -- \\
Post-edit & 0.136 & 0.357 & -- & 0.210 & 0.578 & -- \\
Program & 0.691 & 0.691 & -- & 0.691 & 0.691 & -- \\
\method{} $\beta{=}16$ & 0.012 & 0.024 & 16 & 0.024 & 0.036 & 16 \\
\method{} $\beta{=}\infty$ & \textbf{0.769} & \textbf{0.834} & 72 & \textbf{0.789} & \textbf{0.852} & 52 \\
Full upload & 0.870 & 0.918 & -- & 0.870 & 0.918 & -- \\
\bottomrule
\end{tabular}
\end{table}

With the 0.5B model, the selected tilt ($a=2$, $b=0$) reaches 0.769 strict and 0.834 lenient accuracy, 0.42 above local-only under the lenient score (95\% CI 0.38 to 0.47). With the 1.5B model ($a=1$), it reaches 0.789 and 0.852, 0.18 above local-only. Both end within 0.07 to 0.10 of full upload (0.870 / 0.918). The gain from amplifying the small model's evidence ($a=2$) is consistent with insufficient evidence weighting, as Proposition~\ref{prop:tilt}(b) predicts and the ablation in \S\ref{sec:ablation} confirms. Paired differences and significance for the 0.5B client are reported in Table~\ref{tab:gsm8k} in the Appendix. The answer to RQ1 is thus yes for private values with both local models, and for private documents only on HotpotQA.

\textbf{RQ2: Tilt versus Mixture.} Table~\ref{tab:main} shows that without a budget, the tilt improves on the linear mixture in most settings, and no reported interval establishes an advantage for the mixture. Unbudgeted, with the 0.5B model, the tilt exceeds the mixture on every QA panel by 0.015 to 0.063 F1, significantly on 2WikiMQA/question and MuSiQue/BM25, and on GSM8K by 0.67, where the mixture at full weight only matches the local model (0.100 strict and 0.411 lenient), as Proposition~\ref{prop:tilt}(d) predicts. At $\beta=16$, the QA estimates favour the tilt by 0.009 to 0.045 while it spends fewer nats (6.8 to 12.0 against 11.1 to 13.8), significantly only on MuSiQue/BM25. Tuning the mixture's weight cap on validation for the 0.5B pair ($\tfrac12$ for GSM8K, $\tfrac34$ for QA) does not close the gap. The tuned mixture scores 0.042 on GSM8K, 0.73 below the tilt (95\% CI 0.69 to 0.76), and trails it by 0.022 to 0.046 F1 on QA, with significant gaps on both MuSiQue panels. As the ablation in \S\ref{sec:ablation} shows, giving the mixture the tilt's full reference does not help it on QA either (0.18 against 0.229 for the tilt), so the gap reflects the fusion rule. With 1.5B, the tilt significantly exceeds the mixture on both MuSiQue panels at both budgets; more details are in Appendix~\ref{app:extra}.

\textbf{RQ3: Effect of the Public View.} Table~\ref{tab:main} shows that a more informative view narrows the gap to full upload, mostly by lifting the restricted cloud. From question-only to BM25, the gap narrows from 0.152 to 0.134 F1 on 2WikiMQA and from 0.196 to 0.139 on MuSiQue, while the restricted cloud improves more, by 0.081 and 0.121, so the tilt's advantage over it shrinks. Under the suffix and entity-masked views, the tilt stays within $\pm0.01$ of the restricted cloud, while linear mixtures lose up to 0.09 F1. The tilt helps most when $V$ carries the least evidence; more details for all four views are in Appendix~\ref{app:extra}.

\subsection{Leakage Audit and Budget}
\label{sec:audit}
\begin{table}[t]\small\centering
\caption{Planted-secret audit (96 documents, 48 per question type). \emph{Answer}: rate of the 6-digit code in the answer (local-only sends nothing to the cloud). \emph{Beyond}: rate at which the full code reaches the cloud other than in the answer (uploaded context or drafts), an exposure indicator, not an information measure. F1: utility on the other questions. Bound: Theorem~\ref{thm:ledger}(iii) on guessing the code ($10^{-6}$: blind guess).}
\label{tab:canary}
\begin{tabular}{lrrrrr}
\toprule
 & \multicolumn{2}{c}{Asks for code} & \multicolumn{2}{c}{Other} & \\
Method & Answer & Beyond & Beyond & F1 & Bound \\
\midrule
Restricted cloud & 0.00 & 0.00 & 0.00 & 0.25 & $10^{-6}$ \\
Local only & 0.98 & 0.00 & 0.00 & 0.30 & $10^{-6}$ \\
Full upload & 1.00 & 1.00 & 1.00 & 0.51 & -- \\
Private-draft SD & 0.00 & 0.10 & 0.00 & 0.27 & -- \\
\method{} $\beta{=}4$ & 0.00 & 0.00 & 0.00 & 0.25 & 0.34 \\
\method{} $\beta{=}16$ & 0.00 & 0.00 & 0.00 & 0.29 & 1.00 \\
\method{} $\beta{=}\infty$ & 0.96 & 0.00 & 0.00 & 0.30 & 1.00 \\
Evidence-only tilt & 0.46 & 0.00 & 0.00 & 0.26 & 1.00 \\
Mixture $\beta{=}16$ & 0.00 & 0.00 & 0.00 & 0.31 & 1.00 \\
Mixture $\beta{=}\infty$ & 0.96 & 0.00 & 0.00 & 0.32 & 1.00 \\
\bottomrule
\end{tabular}
\end{table}

\begin{table*}[t]\footnotesize\centering
\caption{One-factor ablations (0.5B local, 7B cloud, a single run). Each row changes one setting of the frozen configuration (GSM8K: $a{=}2,b{=}0$; QA: $a{=}0,b{=}\frac34$; full reference, $c{=}50$, $K{=}64$, $L{=}8$, $\beta{=}\infty$). Cells: level (paired difference from the control, * if its 95\% interval excludes zero); QA pools the four panels of Table~\ref{tab:main}.}
\label{tab:ablation}
\begin{tabular}{llrrrrrrr}
\toprule
 & & \multicolumn{5}{c}{GSM8K / masked numbers} & \multicolumn{2}{c}{Multi-hop QA} \\
\cmidrule(lr){3-7}\cmidrule(lr){8-9}
Factor & Setting & Strict & Lenient & Nats & Req. & Down kB & F1 & Nats \\
\midrule
Control & selected & 0.75 & 0.81 & 71 & 52 & 216 & 0.229 & 17.1 \\
Score & $a=0.5$ & 0.07 (-0.69*) & 0.08 (-0.73*) & 13 & 61 & 253 & -- & -- \\
Score & $a=1$ & 0.66 (-0.09*) & 0.73 (-0.07*) & 57 & 54 & 224 & 0.25 (+0.02) & 5.2 \\
Score & $a=3$ & 0.67 (-0.08*) & 0.78 (-0.02) & 108 & 56 & 231 & -- & -- \\
Score & $a=4$ & 0.62 (-0.13*) & 0.72 (-0.09*) & 153 & 60 & 249 & -- & -- \\
Reference & top-$K$ only & 0.75 (-0.00) & 0.82 (+0.02) & 70 & 52 & 214 & 0.23 (-0.00) & 14.0 \\
Clip & $c=10$ & 0.62 (-0.13*) & 0.67 (-0.14*) & 81 & 57 & 237 & -- & -- \\
Reply size & $K=8$ & 0.72 (-0.03) & 0.76 (-0.05*) & 87 & 56 & 30 & -- & -- \\
Reply size & $K=256$ & 0.76 (+0.01) & 0.82 (+0.01) & 71 & 53 & 862 & -- & -- \\
Lookahead & $L=1$ & 0.75 (+0.00) & 0.81 (+0.00) & 71 & 292 & 151 & -- & -- \\
Lookahead & $L=16$ & 0.75 (+0.00) & 0.81 (+0.00) & 71 & 37 & 304 & -- & -- \\
Budget & $\beta=32$ & 0.06 (-0.69*) & 0.07 (-0.74*) & 32 & 57 & 235 & -- & -- \\
Budget & $\beta=64$ & 0.49 (-0.26*) & 0.53 (-0.28*) & 59 & 53 & 219 & -- & -- \\
Fusion rule & mixture, full ref. & 0.09 (-0.66*) & 0.39 (-0.42*) & 168 & 72 & 296 & 0.18 (-0.05*) & 42.4 \\
\bottomrule
\end{tabular}
\end{table*}

We next check what reaches the cloud and what the ledger costs.

\textbf{Planted-Secret Audit.} Table~\ref{tab:canary} checks Theorem~\ref{thm:answer-only} against the implementation. No answer-only configuration, at any budget and under either fusion rule, places the code in the cloud-received transcript beyond the answer (0/48 on both question types; one-sided 95\% upper bound 6.1\%). 
Re-deriving the client's messages from $V$, the answer and the cloud's recorded replies alone, as in the proof, reproduces all 57{,}372 answer-only transcripts. Private-draft SD never emits the code, but exposes it to the cloud in 5/48 extraction sessions through its drafts, the leakage described by Proposition~\ref{prop:drafts}. Full upload exposes it in every document. 
The ledger binds. Local-only emits the code for 47/48 extraction questions and unbudgeted fusion for 46/48, but the tilt at $\beta\in\{4,16\}$ and the mixture at $\beta=16$ never emit it. 
Bound (iii) is $0.34$ at $\beta=4$ and vacuous at $\beta=16>\log 10^6$, so the zero there is an empirical observation, not a guarantee. On the other questions, the tilt scores 0.29 F1 at $\beta=16$, against 0.25 for restricted and 0.30 for local-only; with 48 questions, neither difference is significant.

\textbf{Ledger Cost.} Table~\ref{tab:gsm8kmain} shows that the ledger makes the cost of the GSM8K gain explicit. At $\beta=16$, the tilt stays near the restricted cloud with both local models. Numeral tokens make up 16\% of a tilted answer but carry 30\% (0.5B) and 43\% (1.5B) of the spend. 

\begin{figure}[t]\centering
\includegraphics[width=0.93\columnwidth]{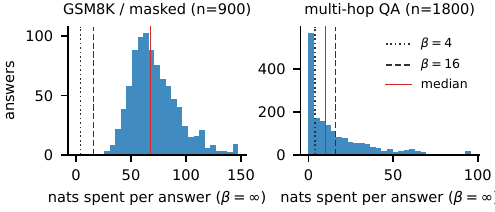}
\caption{Ledger spend per answer for \method{} at $\beta=\infty$ (0.5B local).}
\vspace{-2em}
\label{fig:spend}
\Description{Histograms of nats spent per answer, concentrated around 70 on GSM8K and heavy-tailed near
zero on QA.}
\end{figure}

\textbf{Where the Nats Go.} Figure~\ref{fig:spend} shows that on QA, the spend has a heavy tail. The median answer under the tilt spends 10 nats, and two-thirds spend at most 16, so a budget of 16 nats changes F1 little. Every GSM8K answer exceeds 16 nats (10th percentile 48), so such budgets remove the benefit.

\textbf{Budget Choosing.} A budget should follow from the inference risk the user accepts. For a secret uniform on $m$ values and a tolerated guessing probability $r$, Theorem~\ref{thm:ledger}(iii) gives $\beta\le\mathrm{kl}(r\|1/m)$; for a six-digit code this is 4 nats at $r=0.34$ but 0.006 nats at $r=0.001$, so strong protection leaves the answer little room to use the secret. In our step-by-step GSM8K outputs the charge averages 52 to 72 nats; it bounds disclosure relative to the public reference rather than measuring what an individual answer reveals. Our short QA outputs incurred smaller charges. Table~\ref{tab:main} shows that on MuSiQue with the question-only view, the tilt at $\beta=16$ still beats restricted significantly, by 0.061 with the 0.5B model and by 0.078 with 1.5B, retaining much of the unbudgeted improvement, so a cap of 16 nats keeps the QA benefit, though it is too loose to protect a short secret.

\subsection{Comparison with Zero-Leakage Designs}
\method{} buys its accuracy with ledger spend, so we also compare it with designs that disclose nothing beyond $V$. Table~\ref{tab:gsm8kmain} shows that on GSM8K, a cloud-written program, run by the client on the real numbers, is a strong alternative. It scores 0.691 under both strict and lenient scoring, and a local fallback for failed programs changes almost nothing (0.692). \method{} exceeds it by 0.078 with the 0.5B client (95\% CI 0.031 to 0.127) and by 0.098 with 1.5B (0.051 to 0.147), and by 0.14 and 0.16 under the lenient score. When run offline, the program discloses nothing beyond $V$, whereas \method{} spends 72 and 52 nats for this margin, so a user who cannot accept that spend should prefer the program for arithmetic of this kind. Post-editing is much weaker. Table~\ref{tab:gsm8kmain} shows it below local-only on GSM8K (0.357 against 0.410 lenient with 0.5B), and Table~\ref{tab:main} shows it near the endpoints on QA.

\subsection{Ablation Study}
\label{sec:ablation}
\begin{figure}[t]\centering
\includegraphics[width=0.85\columnwidth]{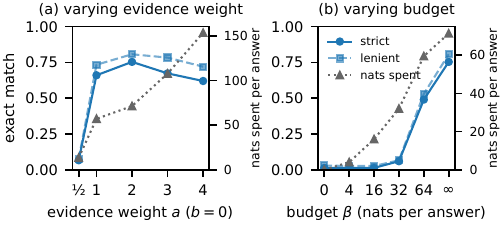}
\caption{GSM8K ablations (0.5B local, a single run): exact match and, dotted on the right axis, mean nats spent per answer. (a) Varying the evidence weight $a$ ($a=2$ is the frozen choice). (b) Varying the per-answer budget $\beta$ ($\beta=0$ is the restricted cloud).}
\label{fig:dose}
\Description{Left: accuracy peaks at evidence weight 2 while spend keeps rising. Right: accuracy is
near zero up to 32 nats, about 0.5 at 64 nats and 0.75 without a budget.}
\end{figure}

\textbf{Evidence Weight.} Table~\ref{tab:ablation} changes one setting at a time against the main run. Figure~\ref{fig:dose}a varies the evidence weight $a$ and shows that accuracy is an inverted U in $a$. At $a=0.5$, the evidence is too weak to move the cloud off its placeholders (0.07). Frozen $a=2$ is best (0.75); $a=1$ loses a significant 0.09, and $a=3$ and $4$ spend more nats (108 and 153 against 71) but score lower.

\textbf{Budget.} Figure~\ref{fig:dose}b varies the per-answer budget $\beta$ and shows that accuracy rises steeply only at large budgets. On GSM8K, 16 nats buy nothing (0.013), 32 nats little (0.060), and 64 nats about two-thirds of the gain, 0.490 against 0.753 unbudgeted.

\textbf{Reference and Clip.} Table~\ref{tab:ablation} shows that restricting the reference to the cloud's top-64 tokens changes accuracy by under 0.01 on GSM8K and QA; a tighter clip ($c=10$) costs 0.13, likely the cause of the evidence-only deficit.

\textbf{Fusion Rule.} Table~\ref{tab:ablation} shows that with the tilt's full reference, the mixture still reaches only 0.09 on GSM8K at 168 nats, against 0.75 at 71 nats for the tilt.

\textbf{Reply Size and Lookahead.} Table~\ref{tab:ablation} shows that with $K=8$, replies shrink sevenfold (30 against 216\,kB) at a cost of 0.03 strict (not significant) and 0.05 lenient accuracy; $K=256$ adds nothing. 
Lookahead changes only the cost. With $L=1$ and $L=16$, all 300 answers are token-identical to the control, while requests fall from 292 to 37. 
A trace-based cost model at a 100\,ms round trip and 10\,Mbit/s puts communication per GSM8K answer at 29.6\,s for $L=1$, the edge-side pattern, 5.4\,s for $L=8$ and 4.0\,s for $L=16$, excluding computation; more details can be found in Appendix~\ref{app:extra}.

\section{Related Work}
\label{sec:related}
\textbf{Speculative and Collaborative Decoding.} Speculative sampling is exact for its target for any proposal~\cite{leviathan2023fast,chen2023acceleratinglargelanguagemodel}. Distributed, cascaded and lossy variants optimise placement, deferral or speed~\cite{yu2025dsddistributedspeculativedecoding,ning2025dssd,bhattacharjee2025conformal,narasimhan2025cascades,holsman2025fuzzy,yin2024theoretical}. SpecSteer~\cite{lv2026specsteersynergizinglocalcontext} and AsymSpec~\cite{liang2026asymspec} let a full-context small model steer a restricted large model; neither analyses what the transcript discloses, which Propositions~\ref{prop:drafts}--\ref{prop:coupling} address. Acceptance patterns also form a side channel~\cite{wei2026speculationspillssecretschannels}.

\textbf{Private LLM Inference.} Prompt sanitisers~\cite{feyisetan2020privacy,tong2025inferdpt} and anonymise-then-restore pipelines such as HaS~\cite{chen2023has}, which our post-edit design adapts, restore a masked answer on the client. CoGenesis~\cite{zhang2024cogenesis} uploads local signals. CoVeil~\cite{zhang2026coveil} reduces exposure empirically without a formal bound; its edge-side message pattern matches restart streaming at $L=1$, as Table~\ref{tab:compare} details.
DP-Fusion~\cite{thareja2026dpfusion} and PMixED~\cite{flemings2024pmixed} bound the divergence of a \emph{single} party's mixed output; \method{} extends this to a two-party Web protocol. 
Cryptographic and DP-in-forward methods are complementary~\cite{li2025privategpt,kei2025shaft,du2023dpforward}.

\section{Conclusion}
\label{sec:limits}
Decoding messages can reveal more than the answer. \method{} makes every message to a Web API a function of the public view and the answer and bounds what the answer itself reveals. Unbudgeted, it brings a small client close to full upload on private values and ahead of a zero-leakage program; strong budgets forfeit most of it.

\textbf{Limitations.} \method{} protects message \emph{content}. Timing can depend on local compute over $X$, which we do not address. $V$ is disclosed by design, and our bounds take $M$ as given. The KL guarantee is average-case and weaker than pure DP; R\'enyi orders $\alpha>1$ give tail bounds at a utility cost. The program baseline requires computation over masked values, which is absent in our QA tasks. Our public views are rule-based and network times are modelled. 

\textbf{Future Work.} This includes zero-leakage baselines for document tasks, learned public views and a live Web API deployment. No personal data was used.

\bibliographystyle{ACM-Reference-Format}
\bibliography{references}


\begin{thebibliography}{37}


\ifx \showCODEN    \undefined \def \showCODEN     #1{\unskip}     \fi
\ifx \showISBNx    \undefined \def \showISBNx     #1{\unskip}     \fi
\ifx \showISBNxiii \undefined \def \showISBNxiii  #1{\unskip}     \fi
\ifx \showISSN     \undefined \def \showISSN      #1{\unskip}     \fi
\ifx \showLCCN     \undefined \def \showLCCN      #1{\unskip}     \fi
\ifx \shownote     \undefined \def \shownote      #1{#1}          \fi
\ifx \showarticletitle \undefined \def \showarticletitle #1{#1}   \fi
\ifx \showURL      \undefined \def \showURL       {\relax}        \fi
\providecommand\bibfield[2]{#2}
\providecommand\bibinfo[2]{#2}
\providecommand\natexlab[1]{#1}
\providecommand\showeprint[2][]{arXiv:#2}

\bibitem[Bai et~al\mbox{.}(2024)]%
        {bai2024longbench}
\bibfield{author}{\bibinfo{person}{Yushi Bai}, \bibinfo{person}{Xin Lv},
  \bibinfo{person}{Jiajie Zhang}, \bibinfo{person}{Hongchang Lyu},
  \bibinfo{person}{Jiankai Tang}, \bibinfo{person}{Zhidian Huang},
  \bibinfo{person}{Zhengxiao Du}, \bibinfo{person}{Xiao Liu},
  \bibinfo{person}{Aohan Zeng}, \bibinfo{person}{Lei Hou},
  \bibinfo{person}{Yuxiao Dong}, \bibinfo{person}{Jie Tang}, {and}
  \bibinfo{person}{Juanzi Li}.} \bibinfo{year}{2024}\natexlab{}.
\newblock \showarticletitle{{LongBench}: A Bilingual, Multitask Benchmark for
  Long Context Understanding}. In \bibinfo{booktitle}{\emph{Proceedings of the
  62nd annual meeting of the association for computational linguistics (volume
  1: Long papers)}}. \bibinfo{pages}{3119--3137}.
\newblock
\href{https://doi.org/10.18653/v1/2024.acl-long.172}{doi:\nolinkurl{10.18653/v1/2024.acl-long.172}}


\bibitem[Bhattacharjee et~al\mbox{.}(2025)]%
        {bhattacharjee2025conformal}
\bibfield{author}{\bibinfo{person}{Payel Bhattacharjee},
  \bibinfo{person}{Fengwei Tian}, \bibinfo{person}{Meiyu Zhong},
  \bibinfo{person}{Guangyi Zhang}, \bibinfo{person}{Osvaldo Simeone}, {and}
  \bibinfo{person}{Ravi Tandon}.} \bibinfo{year}{2025}\natexlab{}.
\newblock \bibinfo{title}{Conformal Sparsification for Bandwidth-Efficient
  Edge-Cloud Speculative Decoding}.
\newblock
\showeprint[arxiv]{2510.09942}~[cs.LG]
\urldef\tempurl%
\url{https://arxiv.org/abs/2510.09942}
\showURL{%
\tempurl}


\bibitem[Chen et~al\mbox{.}(2023a)]%
        {chen2023acceleratinglargelanguagemodel}
\bibfield{author}{\bibinfo{person}{Charlie Chen}, \bibinfo{person}{Sebastian
  Borgeaud}, \bibinfo{person}{Geoffrey Irving}, \bibinfo{person}{Jean-Baptiste
  Lespiau}, \bibinfo{person}{Laurent Sifre}, {and} \bibinfo{person}{John
  Jumper}.} \bibinfo{year}{2023}\natexlab{a}.
\newblock \bibinfo{title}{Accelerating Large Language Model Decoding with
  Speculative Sampling}.
\newblock
\showeprint[arxiv]{2302.01318}~[cs.CL]
\urldef\tempurl%
\url{https://arxiv.org/abs/2302.01318}
\showURL{%
\tempurl}


\bibitem[Chen et~al\mbox{.}(2023b)]%
        {chen2023has}
\bibfield{author}{\bibinfo{person}{Yu Chen}, \bibinfo{person}{Tingxin Li},
  \bibinfo{person}{Huiming Liu}, {and} \bibinfo{person}{Yang Yu}.}
  \bibinfo{year}{2023}\natexlab{b}.
\newblock \bibinfo{title}{Hide and Seek ({HaS}): A Lightweight Framework for
  Prompt Privacy Protection}.
\newblock
\showeprint[arxiv]{2309.03057}~[cs.CR]
\urldef\tempurl%
\url{https://arxiv.org/abs/2309.03057}
\showURL{%
\tempurl}


\bibitem[Cobbe et~al\mbox{.}(2021)]%
        {cobbe2021trainingverifierssolvemath}
\bibfield{author}{\bibinfo{person}{Karl Cobbe}, \bibinfo{person}{Vineet
  Kosaraju}, \bibinfo{person}{Mohammad Bavarian}, \bibinfo{person}{Mark Chen},
  \bibinfo{person}{Heewoo Jun}, \bibinfo{person}{Lukasz Kaiser},
  \bibinfo{person}{Matthias Plappert}, \bibinfo{person}{Jerry Tworek},
  \bibinfo{person}{Jacob Hilton}, \bibinfo{person}{Reiichiro Nakano},
  \bibinfo{person}{Christopher Hesse}, {and} \bibinfo{person}{John Schulman}.}
  \bibinfo{year}{2021}\natexlab{}.
\newblock \bibinfo{title}{Training Verifiers to Solve Math Word Problems}.
\newblock
\showeprint[arxiv]{2110.14168}~[cs.LG]
\urldef\tempurl%
\url{https://arxiv.org/abs/2110.14168}
\showURL{%
\tempurl}


\bibitem[Cover and Thomas(2006)]%
        {cover2006elements}
\bibfield{author}{\bibinfo{person}{Thomas~M. Cover} {and}
  \bibinfo{person}{Joy~A. Thomas}.} \bibinfo{year}{2006}\natexlab{}.
\newblock \bibinfo{booktitle}{\emph{Elements of Information Theory}
  (\bibinfo{edition}{2nd} ed.)}.
\newblock \bibinfo{publisher}{Wiley}.
\newblock


\bibitem[Csisz{\'a}r(1975)]%
        {csiszar1975idivergence}
\bibfield{author}{\bibinfo{person}{Imre Csisz{\'a}r}.}
  \bibinfo{year}{1975}\natexlab{}.
\newblock \showarticletitle{{I}-Divergence Geometry of Probability
  Distributions and Minimization Problems}.
\newblock \bibinfo{journal}{\emph{The Annals of Probability}}
  \bibinfo{volume}{3}, \bibinfo{number}{1} (\bibinfo{year}{1975}),
  \bibinfo{pages}{146--158}.
\newblock
\href{https://doi.org/10.1214/aop/1176996454}{doi:\nolinkurl{10.1214/aop/1176996454}}


\bibitem[Du et~al\mbox{.}(2023)]%
        {du2023dpforward}
\bibfield{author}{\bibinfo{person}{Minxin Du}, \bibinfo{person}{Xiang Yue},
  \bibinfo{person}{Sherman S.~M. Chow}, \bibinfo{person}{Tianhao Wang},
  \bibinfo{person}{Chenyu Huang}, {and} \bibinfo{person}{Huan Sun}.}
  \bibinfo{year}{2023}\natexlab{}.
\newblock \showarticletitle{{DP-Forward}: Fine-tuning and Inference on Language
  Models with Differential Privacy in Forward Pass}. In
  \bibinfo{booktitle}{\emph{Proceedings of the 2023 ACM SIGSAC Conference on
  Computer and Communications Security}}. \bibinfo{publisher}{Association for
  Computing Machinery}, \bibinfo{pages}{2665--2679}.
\newblock
\href{https://doi.org/10.1145/3576915.3616592}{doi:\nolinkurl{10.1145/3576915.3616592}}


\bibitem[Feldman and Zrnic(2021)]%
        {feldman2021filter}
\bibfield{author}{\bibinfo{person}{Vitaly Feldman} {and}
  \bibinfo{person}{Tijana Zrnic}.} \bibinfo{year}{2021}\natexlab{}.
\newblock \showarticletitle{Individual Privacy Accounting via a R{\'e}nyi
  Filter}. In \bibinfo{booktitle}{\emph{Advances in Neural Information
  Processing Systems}}, Vol.~\bibinfo{volume}{34}.
  \bibinfo{pages}{28080--28091}.
\newblock


\bibitem[Feyisetan et~al\mbox{.}(2020)]%
        {feyisetan2020privacy}
\bibfield{author}{\bibinfo{person}{Oluwaseyi Feyisetan}, \bibinfo{person}{Borja
  Balle}, \bibinfo{person}{Thomas Drake}, {and} \bibinfo{person}{Tom Diethe}.}
  \bibinfo{year}{2020}\natexlab{}.
\newblock \showarticletitle{Privacy- and Utility-Preserving Textual Analysis
  via Calibrated Multivariate Perturbations}. In
  \bibinfo{booktitle}{\emph{Proceedings of the 13th international conference on
  web search and data mining}}. \bibinfo{pages}{178--186}.
\newblock
\href{https://doi.org/10.1145/3336191.3371856}{doi:\nolinkurl{10.1145/3336191.3371856}}


\bibitem[Flemings et~al\mbox{.}(2024)]%
        {flemings2024pmixed}
\bibfield{author}{\bibinfo{person}{James Flemings}, \bibinfo{person}{Meisam
  Razaviyayn}, {and} \bibinfo{person}{Murali Annavaram}.}
  \bibinfo{year}{2024}\natexlab{}.
\newblock \showarticletitle{Differentially private next-token prediction of
  large language models}. In \bibinfo{booktitle}{\emph{Proceedings of the 2024
  Conference of the North American Chapter of the Association for Computational
  Linguistics: Human Language Technologies (Volume 1: Long Papers)}}.
  \bibinfo{pages}{4390--4404}.
\newblock


\bibitem[Gao et~al\mbox{.}(2023)]%
        {gao2023pal}
\bibfield{author}{\bibinfo{person}{Luyu Gao}, \bibinfo{person}{Aman Madaan},
  \bibinfo{person}{Shuyan Zhou}, \bibinfo{person}{Uri Alon},
  \bibinfo{person}{Pengfei Liu}, \bibinfo{person}{Yiming Yang},
  \bibinfo{person}{Jamie Callan}, {and} \bibinfo{person}{Graham Neubig}.}
  \bibinfo{year}{2023}\natexlab{}.
\newblock \showarticletitle{{PAL}: Program-aided Language Models}. In
  \bibinfo{booktitle}{\emph{Proceedings of the 40th International Conference on
  Machine Learning}} \emph{(\bibinfo{series}{Proceedings of Machine Learning
  Research}, Vol.~\bibinfo{volume}{202})}. \bibinfo{publisher}{PMLR},
  \bibinfo{pages}{10764--10799}.
\newblock


\bibitem[Ho et~al\mbox{.}(2020)]%
        {ho2020wikimqa}
\bibfield{author}{\bibinfo{person}{Xanh Ho}, \bibinfo{person}{Anh-Khoa
  Duong~Nguyen}, \bibinfo{person}{Saku Sugawara}, {and} \bibinfo{person}{Akiko
  Aizawa}.} \bibinfo{year}{2020}\natexlab{}.
\newblock \showarticletitle{Constructing A Multi-hop {QA} Dataset for
  Comprehensive Evaluation of Reasoning Steps}. In
  \bibinfo{booktitle}{\emph{Proceedings of the 28th International Conference on
  Computational Linguistics}}. \bibinfo{pages}{6609--6625}.
\newblock
\href{https://doi.org/10.18653/v1/2020.coling-main.580}{doi:\nolinkurl{10.18653/v1/2020.coling-main.580}}


\bibitem[Holsman et~al\mbox{.}(2025)]%
        {holsman2025fuzzy}
\bibfield{author}{\bibinfo{person}{Maximilian Holsman}, \bibinfo{person}{Yukun
  Huang}, {and} \bibinfo{person}{Bhuwan Dhingra}.}
  \bibinfo{year}{2025}\natexlab{}.
\newblock \showarticletitle{Fuzzy Speculative Decoding for a Tunable
  Accuracy-Runtime Tradeoff}. In \bibinfo{booktitle}{\emph{Findings of the
  Association for Computational Linguistics: ACL 2025}}.
  \bibinfo{publisher}{Association for Computational Linguistics},
  \bibinfo{pages}{26257--26273}.
\newblock
\href{https://doi.org/10.18653/v1/2025.findings-acl.1346}{doi:\nolinkurl{10.18653/v1/2025.findings-acl.1346}}


\bibitem[Kei and Chow(2025)]%
        {kei2025shaft}
\bibfield{author}{\bibinfo{person}{Andes Y.~L. Kei} {and}
  \bibinfo{person}{Sherman S.~M. Chow}.} \bibinfo{year}{2025}\natexlab{}.
\newblock \showarticletitle{{SHAFT}: Secure, Handy, Accurate, and Fast
  Transformer Inference}. In \bibinfo{booktitle}{\emph{Network and Distributed
  System Security Symposium}}. \bibinfo{publisher}{Internet Society}.
\newblock
\urldef\tempurl%
\url{https://www.ndss-symposium.org/ndss-paper/shaft-secure-handy-accurate-and-fast-transformer-inference/}
\showURL{%
\tempurl}


\bibitem[Leviathan et~al\mbox{.}(2023)]%
        {leviathan2023fast}
\bibfield{author}{\bibinfo{person}{Yaniv Leviathan}, \bibinfo{person}{Matan
  Kalman}, {and} \bibinfo{person}{Yossi Matias}.}
  \bibinfo{year}{2023}\natexlab{}.
\newblock \showarticletitle{Fast inference from transformers via speculative
  decoding}. In \bibinfo{booktitle}{\emph{International Conference on Machine
  Learning}}, Vol.~\bibinfo{volume}{202}. PMLR, \bibinfo{pages}{19274--19286}.
\newblock
\urldef\tempurl%
\url{https://proceedings.mlr.press/v202/leviathan23a.html}
\showURL{%
\tempurl}


\bibitem[Li et~al\mbox{.}(2025)]%
        {li2025privategpt}
\bibfield{author}{\bibinfo{person}{Zhengyi Li}, \bibinfo{person}{Yue Guan},
  \bibinfo{person}{Kang Yang}, \bibinfo{person}{Yu Feng}, \bibinfo{person}{Ning
  Liu}, \bibinfo{person}{Yu Yu}, \bibinfo{person}{Jingwen Leng}, {and}
  \bibinfo{person}{Minyi Guo}.} \bibinfo{year}{2025}\natexlab{}.
\newblock \showarticletitle{An Efficient Private {GPT} Never Autoregressively
  Decodes}. In \bibinfo{booktitle}{\emph{Proceedings of the 42nd International
  Conference on Machine Learning}} \emph{(\bibinfo{series}{Proceedings of
  Machine Learning Research}, Vol.~\bibinfo{volume}{267})}.
  \bibinfo{publisher}{PMLR}, \bibinfo{pages}{34410--34428}.
\newblock
\urldef\tempurl%
\url{https://proceedings.mlr.press/v267/li25q.html}
\showURL{%
\tempurl}


\bibitem[Liang et~al\mbox{.}(2026)]%
        {liang2026asymspec}
\bibfield{author}{\bibinfo{person}{Sheng Liang}, \bibinfo{person}{Yongyue
  Zhang}, \bibinfo{person}{Nathanael Brian}, \bibinfo{person}{Hang Lv},
  \bibinfo{person}{Hao Wang}, \bibinfo{person}{Chen Zhang}, {and}
  \bibinfo{person}{Yong Liu}.} \bibinfo{year}{2026}\natexlab{}.
\newblock \bibinfo{title}{{AsymSpec}: Context-Asymmetric Speculative Decoding
  for Agentic {LLMs}}.
\newblock
\showeprint[arxiv]{2608.26004}
\urldef\tempurl%
\url{https://arxiv.org/abs/2608.26004}
\showURL{%
\tempurl}


\bibitem[Liu et~al\mbox{.}(2024)]%
        {liu2024proxy}
\bibfield{author}{\bibinfo{person}{Alisa Liu}, \bibinfo{person}{Xiaochuang
  Han}, \bibinfo{person}{Yizhong Wang}, \bibinfo{person}{Yulia Tsvetkov},
  \bibinfo{person}{Yejin Choi}, {and} \bibinfo{person}{Noah~A. Smith}.}
  \bibinfo{year}{2024}\natexlab{}.
\newblock \showarticletitle{Tuning Language Models by Proxy}. In
  \bibinfo{booktitle}{\emph{Conference on Language Modeling}}.
\newblock
\urldef\tempurl%
\url{https://arxiv.org/abs/2401.08565}
\showURL{%
\tempurl}


\bibitem[Lv et~al\mbox{.}(2026)]%
        {lv2026specsteersynergizinglocalcontext}
\bibfield{author}{\bibinfo{person}{Hang Lv}, \bibinfo{person}{Sheng Liang},
  \bibinfo{person}{Hao Wang}, \bibinfo{person}{Yongyue Zhang},
  \bibinfo{person}{Hongchao Gu}, \bibinfo{person}{Wei Guo},
  \bibinfo{person}{Defu Lian}, \bibinfo{person}{Yong Liu}, {and}
  \bibinfo{person}{Enhong Chen}.} \bibinfo{year}{2026}\natexlab{}.
\newblock \bibinfo{title}{SpecSteer: Synergizing Local Context and Global
  Reasoning for Efficient Personalized Generation}.
\newblock
\showeprint[arxiv]{2603.16219}~[cs.CL]
\urldef\tempurl%
\url{https://arxiv.org/abs/2603.16219}
\showURL{%
\tempurl}


\bibitem[Mironov(2017)]%
        {mironov2017renyi}
\bibfield{author}{\bibinfo{person}{Ilya Mironov}.}
  \bibinfo{year}{2017}\natexlab{}.
\newblock \showarticletitle{R{\'e}nyi Differential Privacy}. In
  \bibinfo{booktitle}{\emph{IEEE 30th Computer Security Foundations Symposium
  (CSF)}}. \bibinfo{pages}{263--275}.
\newblock
\href{https://doi.org/10.1109/CSF.2017.11}{doi:\nolinkurl{10.1109/CSF.2017.11}}


\bibitem[Narasimhan et~al\mbox{.}(2025)]%
        {narasimhan2025cascades}
\bibfield{author}{\bibinfo{person}{Harikrishna Narasimhan},
  \bibinfo{person}{Wittawat Jitkrittum}, \bibinfo{person}{Ankit~Singh Rawat},
  \bibinfo{person}{Seungyeon Kim}, \bibinfo{person}{Neha Gupta},
  \bibinfo{person}{Aditya~Krishna Menon}, {and} \bibinfo{person}{Sanjiv
  Kumar}.} \bibinfo{year}{2025}\natexlab{}.
\newblock \showarticletitle{Faster Cascades via Speculative Decoding}. In
  \bibinfo{booktitle}{\emph{The Thirteenth International Conference on Learning
  Representations}}.
\newblock
\urldef\tempurl%
\url{https://proceedings.iclr.cc/paper_files/paper/2025/hash/6f43166f50f26e8d8f3edc5545b0749f-Abstract-Conference.html}
\showURL{%
\tempurl}


\bibitem[Ning et~al\mbox{.}(2025)]%
        {ning2025dssd}
\bibfield{author}{\bibinfo{person}{Jiahong Ning}, \bibinfo{person}{Ce Zheng},
  {and} \bibinfo{person}{Tingting Yang}.} \bibinfo{year}{2025}\natexlab{}.
\newblock \showarticletitle{DSSD: Efficient Edge-Device LLM Deployment and
  Collaborative Inference via Distributed Split Speculative Decoding}. In
  \bibinfo{booktitle}{\emph{ICML 2025 Workshop on Machine Learning for Wireless
  Communication and Networks (ML4Wireless)}}.
\newblock
\urldef\tempurl%
\url{https://arxiv.org/abs/2507.12000}
\showURL{%
\tempurl}


\bibitem[Robertson and Zaragoza(2009)]%
        {robertson2009bm25}
\bibfield{author}{\bibinfo{person}{Stephen Robertson} {and}
  \bibinfo{person}{Hugo Zaragoza}.} \bibinfo{year}{2009}\natexlab{}.
\newblock \showarticletitle{The Probabilistic Relevance Framework: {BM25} and
  Beyond}.
\newblock \bibinfo{journal}{\emph{Foundations and Trends in Information
  Retrieval}} \bibinfo{volume}{3}, \bibinfo{number}{4} (\bibinfo{year}{2009}),
  \bibinfo{pages}{333--389}.
\newblock
\href{https://doi.org/10.1561/1500000019}{doi:\nolinkurl{10.1561/1500000019}}


\bibitem[Robertson et~al\mbox{.}(1994)]%
        {robertson1994okapi}
\bibfield{author}{\bibinfo{person}{Stephen~E. Robertson},
  \bibinfo{person}{Steve Walker}, \bibinfo{person}{Susan Jones},
  \bibinfo{person}{Micheline~M. Hancock-Beaulieu}, {and} \bibinfo{person}{Mike
  Gatford}.} \bibinfo{year}{1994}\natexlab{}.
\newblock \showarticletitle{Okapi at {TREC}-3}. In
  \bibinfo{booktitle}{\emph{Proceedings of the Third Text REtrieval Conference
  (TREC-3)}} \emph{(\bibinfo{series}{NIST Special Publication},
  Vol.~\bibinfo{volume}{500-225})}. \bibinfo{publisher}{National Institute of
  Standards and Technology}, \bibinfo{address}{Gaithersburg, MD, USA},
  \bibinfo{pages}{109--126}.
\newblock


\bibitem[Shi et~al\mbox{.}(2024)]%
        {shi2024cad}
\bibfield{author}{\bibinfo{person}{Weijia Shi}, \bibinfo{person}{Xiaochuang
  Han}, \bibinfo{person}{Mike Lewis}, \bibinfo{person}{Yulia Tsvetkov},
  \bibinfo{person}{Luke Zettlemoyer}, {and} \bibinfo{person}{Wen-tau Yih}.}
  \bibinfo{year}{2024}\natexlab{}.
\newblock \showarticletitle{Trusting your evidence: Hallucinate less with
  context-aware decoding}. In \bibinfo{booktitle}{\emph{Proceedings of the 2024
  Conference of the North American Chapter of the Association for Computational
  Linguistics: Human Language Technologies (Volume 2: Short Papers)}}.
  \bibinfo{pages}{783--791}.
\newblock


\bibitem[Thakur et~al\mbox{.}(2021)]%
        {thakur2021beir}
\bibfield{author}{\bibinfo{person}{Nandan Thakur}, \bibinfo{person}{Nils
  Reimers}, \bibinfo{person}{Andreas R{\"u}ckl{\'e}}, \bibinfo{person}{Abhishek
  Srivastava}, {and} \bibinfo{person}{Iryna Gurevych}.}
  \bibinfo{year}{2021}\natexlab{}.
\newblock \showarticletitle{{BEIR}: A Heterogeneous Benchmark for Zero-shot
  Evaluation of Information Retrieval Models}. In
  \bibinfo{booktitle}{\emph{Proceedings of the Neural Information Processing
  Systems Track on Datasets and Benchmarks}}, Vol.~\bibinfo{volume}{1}.
\newblock
\urldef\tempurl%
\url{https://datasets-benchmarks-proceedings.neurips.cc/paper/2021/hash/65b9eea6e1cc6bb9f0cd2a47751a186f-Abstract-round2.html}
\showURL{%
\tempurl}


\bibitem[Thareja et~al\mbox{.}(2026)]%
        {thareja2026dpfusion}
\bibfield{author}{\bibinfo{person}{Rushil Thareja}, \bibinfo{person}{Preslav
  Nakov}, \bibinfo{person}{Praneeth Vepakomma}, {and} \bibinfo{person}{Nils
  Lukas}.} \bibinfo{year}{2026}\natexlab{}.
\newblock \showarticletitle{DP-fusion: Token-level differentially private
  inference for large language models}. In
  \bibinfo{booktitle}{\emph{International Conference on Learning
  Representations}}, Vol.~\bibinfo{volume}{2026}.
  \bibinfo{pages}{151088--151120}.
\newblock


\bibitem[Tong et~al\mbox{.}(2025)]%
        {tong2025inferdpt}
\bibfield{author}{\bibinfo{person}{Meng Tong}, \bibinfo{person}{Kejiang Chen},
  \bibinfo{person}{Jie Zhang}, \bibinfo{person}{Yuang Qi},
  \bibinfo{person}{Weiming Zhang}, \bibinfo{person}{Nenghai Yu},
  \bibinfo{person}{Tianwei Zhang}, {and} \bibinfo{person}{Zhikun Zhang}.}
  \bibinfo{year}{2025}\natexlab{}.
\newblock \showarticletitle{Inferdpt: Privacy-preserving inference for
  closed-box large language models}.
\newblock \bibinfo{journal}{\emph{IEEE Transactions on Dependable and Secure
  Computing}} \bibinfo{volume}{22}, \bibinfo{number}{5} (\bibinfo{year}{2025}),
  \bibinfo{pages}{4625--4640}.
\newblock


\bibitem[Trivedi et~al\mbox{.}(2022)]%
        {trivedi2022musique}
\bibfield{author}{\bibinfo{person}{Harsh Trivedi}, \bibinfo{person}{Niranjan
  Balasubramanian}, \bibinfo{person}{Tushar Khot}, {and}
  \bibinfo{person}{Ashish Sabharwal}.} \bibinfo{year}{2022}\natexlab{}.
\newblock \showarticletitle{{MuSiQue}: Multihop Questions via Single-hop
  Question Composition}.
\newblock \bibinfo{journal}{\emph{Transactions of the Association for
  Computational Linguistics}}  \bibinfo{volume}{10} (\bibinfo{year}{2022}),
  \bibinfo{pages}{539--554}.
\newblock
\href{https://doi.org/10.1162/tacl_a_00475}{doi:\nolinkurl{10.1162/tacl_a_00475}}


\bibitem[Wei et~al\mbox{.}(2026)]%
        {wei2026speculationspillssecretschannels}
\bibfield{author}{\bibinfo{person}{Jiankun Wei}, \bibinfo{person}{Abdulrahman
  Abdulrazzag}, \bibinfo{person}{Tianchen Zhang}, \bibinfo{person}{Adel
  Muursepp}, {and} \bibinfo{person}{Gururaj Saileshwar}.}
  \bibinfo{year}{2026}\natexlab{}.
\newblock \bibinfo{title}{When Speculation Spills Secrets: Side Channels via
  Speculative Decoding In LLMs}.
\newblock
\showeprint[arxiv]{2411.01076}~[cs.CL]
\urldef\tempurl%
\url{https://arxiv.org/abs/2411.01076}
\showURL{%
\tempurl}


\bibitem[Yang et~al\mbox{.}(2024)]%
        {qwen2024qwen25}
\bibfield{author}{\bibinfo{person}{An Yang} {et~al\mbox{.}}}
  \bibinfo{year}{2024}\natexlab{}.
\newblock \bibinfo{title}{{Qwen2.5} Technical Report}.
\newblock
\showeprint[arxiv]{2412.15115}
\urldef\tempurl%
\url{https://arxiv.org/abs/2412.15115}
\showURL{%
\tempurl}


\bibitem[Yang et~al\mbox{.}(2018)]%
        {yang2018hotpotqa}
\bibfield{author}{\bibinfo{person}{Zhilin Yang}, \bibinfo{person}{Peng Qi},
  \bibinfo{person}{Saizheng Zhang}, \bibinfo{person}{Yoshua Bengio},
  \bibinfo{person}{William Cohen}, \bibinfo{person}{Ruslan Salakhutdinov},
  {and} \bibinfo{person}{Christopher~D. Manning}.}
  \bibinfo{year}{2018}\natexlab{}.
\newblock \showarticletitle{{HotpotQA}: A Dataset for Diverse, Explainable
  Multi-hop Question Answering}. In \bibinfo{booktitle}{\emph{Proceedings of
  the 2018 Conference on Empirical Methods in Natural Language Processing}}.
  \bibinfo{pages}{2369--2380}.
\newblock
\href{https://doi.org/10.18653/v1/D18-1259}{doi:\nolinkurl{10.18653/v1/D18-1259}}


\bibitem[Yin et~al\mbox{.}(2024)]%
        {yin2024theoretical}
\bibfield{author}{\bibinfo{person}{Ming Yin}, \bibinfo{person}{Minshuo Chen},
  \bibinfo{person}{Kaixuan Huang}, {and} \bibinfo{person}{Mengdi Wang}.}
  \bibinfo{year}{2024}\natexlab{}.
\newblock \showarticletitle{A Theoretical Perspective for Speculative Decoding
  Algorithm}. In \bibinfo{booktitle}{\emph{Advances in Neural Information
  Processing Systems}}, Vol.~\bibinfo{volume}{37}.
\newblock
\urldef\tempurl%
\url{https://proceedings.neurips.cc/paper_files/paper/2024/hash/e7349e785900b93d8b4971a3f2c1cefe-Abstract-Conference.html}
\showURL{%
\tempurl}


\bibitem[Yu et~al\mbox{.}(2025)]%
        {yu2025dsddistributedspeculativedecoding}
\bibfield{author}{\bibinfo{person}{Fengze Yu}, \bibinfo{person}{Leshu Li},
  \bibinfo{person}{Brad McDanel}, {and} \bibinfo{person}{Sai~Qian Zhang}.}
  \bibinfo{year}{2025}\natexlab{}.
\newblock \bibinfo{title}{DSD: A Distributed Speculative Decoding Solution for
  Edge-Cloud Agile Large Model Serving}.
\newblock
\showeprint[arxiv]{2511.21669}~[cs.LG]
\urldef\tempurl%
\url{https://arxiv.org/abs/2511.21669}
\showURL{%
\tempurl}


\bibitem[Zhang et~al\mbox{.}(2024)]%
        {zhang2024cogenesis}
\bibfield{author}{\bibinfo{person}{Kaiyan Zhang}, \bibinfo{person}{Jianyu
  Wang}, \bibinfo{person}{Ermo Hua}, \bibinfo{person}{Biqing Qi},
  \bibinfo{person}{Ning Ding}, {and} \bibinfo{person}{Bowen Zhou}.}
  \bibinfo{year}{2024}\natexlab{}.
\newblock \showarticletitle{{CoGenesis}: A Framework Collaborating Large and
  Small Language Models for Secure Context-Aware Instruction Following}. In
  \bibinfo{booktitle}{\emph{Proceedings of the 62nd Annual Meeting of the
  Association for Computational Linguistics (Volume 1: Long Papers)}}.
  \bibinfo{pages}{4295--4312}.
\newblock
\href{https://doi.org/10.18653/v1/2024.acl-long.235}{doi:\nolinkurl{10.18653/v1/2024.acl-long.235}}


\bibitem[Zhang et~al\mbox{.}(2026)]%
        {zhang2026coveil}
\bibfield{author}{\bibinfo{person}{Kejia Zhang}, \bibinfo{person}{Tianyuan
  Zou}, \bibinfo{person}{Zixuan Gu}, {and} \bibinfo{person}{Yang Liu}.}
  \bibinfo{year}{2026}\natexlab{}.
\newblock \bibinfo{title}{Auditing and Mitigating Privacy Leakage in Cloud-Edge
  Collaborative Decoding}.
\newblock
\showeprint[arxiv]{2608.29111}
\urldef\tempurl%
\url{https://arxiv.org/abs/2608.29111}
\showURL{%
\tempurl}


\end{thebibliography}

\appendix
\renewcommand{\dbltopfraction}{0.85}
\section{Proofs}
\label{app:proofs}
\textbf{Proof of Proposition~\ref{prop:drafts}.} Take any prior on $X\in\{0,1\}$, a fixed public view $V$, a one-token answer and three distinct tokens $a_0,a_1,b$. Let the local model draft deterministically, $q(\cdot\mid X)=\delta_{a_X}$, and let the restricted cloud be $p=\delta_b$. The client sends the draft $a_X$. The cloud accepts it with probability $\min\{1,p(a_X)/q(a_X)\}=0$ and samples the correction from the residual, which is proportional to $(p-q)_+=\delta_b$. Hence $Y=b$ for every $X$, so the output law is exactly $p$ and $Y$ is independent of $X$. The transcript contains the draft $a_X$, which determines $X$, so $H(X\mid\T,V,Y)=0$. As $(V,Y)$ is constant, $I(X;\T\mid V,Y)=H(X\mid V,Y)-H(X\mid\T,V,Y)=H(X)$.

\textbf{Proof of Proposition~\ref{prop:coupling}.} Under the reference, $\mu=p$, so every draft is accepted. Under $X$, $\Pr[Y=y,D=y]=\min\{\mu(y),p(y)\}$ and $\Pr[Y=y]=\mu(y)$, so the rejection probability is \[\textstyle\sum_d p(d)\bigl(1-\min\{1,\mu(d)/p(d)\}\bigr)=\TV(\mu,p).\] The event $\{Y\neq D\}$ has reference probability $0$ and positive probability under $X$, so the divergence of the law of $(D,Y)$ is $\infty$ for all $\alpha\ge1$; for the additive bound, couple the two processes on the event that no rejection occurs and apply a union bound over positions. The residual is supported on $\{\mu>p\}$, while a draft $d$ is rejected only if $\mu(d)<p(d)$, so a rejected draft is never emitted, and acceptance equals $\mathbf 1\{Y=D\}$.

\textbf{Proof of Theorem~\ref{thm:answer-only}.} The cloud's $i$-th reply is a function of $V$, $R$ and the client messages so far. The client's $(i{+}1)$-th message is the prefix $Y_{1:j}$, where $j$ is the first position after the previous prefix at which $Y_j$ differs from the streamed token or the lookahead ends. Given $(V,R,Y)$, this recursion determines every message.

\textbf{Proof of Lemma~\ref{lem:unique}.} With $\pi=\delta_d$, $\sum_{d'}\pi(d')K(y\mid d')=K(y\mid d)$, which must equal $\mu(y)$. Maximal coupling accepts $d$ with probability $\pi(d)\min\{1,\mu(d)/\pi(d)\}=\mu(d)$.

The ledger's bisection and Proposition~\ref{prop:tilt}(a) rest on the following lemma.

\begin{lemma}[Monotone charge]\label{lem:monotone}
For $\alpha\ge1$, $D_\alpha(\mu^\lambda\|\tilde p)$ is non-decreasing in $\lambda\in[0,1]$; for $\alpha=1$ it is strictly increasing when $s$ is not constant on the support of $\tilde p$. The linear mixture $r_\lambda=\lambda q+(1-\lambda)\tilde p$ satisfies the same, non-strictly.
\end{lemma}

\textbf{Proof of Lemma~\ref{lem:monotone}.} Write $p$ for $\tilde p$ and $\psi(\lambda)=\log\E_p e^{\lambda s}$, so that $\mu^\lambda=p\,e^{\lambda s-\psi(\lambda)}$, $\psi'(\lambda)=\E_{\mu^\lambda}s$ and $\psi''(\lambda)=\operatorname{Var}_{\mu^\lambda}(s)$. For $\alpha=1$, $\KL(\mu^\lambda\|p)=\lambda\psi'(\lambda)-\psi(\lambda)$, whose derivative $\lambda\operatorname{Var}_{\mu^\lambda}(s)$ is non-negative, and positive for $\lambda>0$ unless $s$ is constant on the support of $p$. For $\alpha>1$, $\E_p[(\mu^\lambda/p)^\alpha]=e^{\psi(\alpha\lambda)-\alpha\psi(\lambda)}$, so $D_\alpha(\mu^\lambda\|p)=[\psi(\alpha\lambda)-\alpha\psi(\lambda)]/(\alpha-1)$, whose derivative $\frac{\alpha}{\alpha-1}[\psi'(\alpha\lambda)-\psi'(\lambda)]$ is non-negative because $\psi$ is convex and $\alpha\lambda\ge\lambda$. For the mixture $r_\lambda=\lambda q+(1-\lambda)p$ and $\alpha>1$, the function
$g(\lambda)=\sum_xp(x)\bigl(1+\lambda(q(x)/p(x)-1)\bigr)^{\alpha}$ is convex, as a sum of convex functions of
an affine argument, with $g'(0)=\alpha\sum_x(q(x)-p(x))=0$; hence $g$, and with it
$D_\alpha(r_\lambda\|p)=\log g(\lambda)/(\alpha-1)$, is non-decreasing on $[0,1]$. For $\alpha=1$,
$\KL(r_\lambda\|p)$ is convex in $\lambda$ and zero at $0$. The R\'enyi divergence itself need not be convex in $\lambda$, but monotonicity suffices for the bisection.

\textbf{Proof of Proposition~\ref{prop:tilt}.} (a) By Lemma~\ref{lem:monotone}, $\KL(\mu^\lambda\|\tilde p)$ is continuous and strictly increasing on $[0,1]$, from $0$ to $\KL(\mu^1\|\tilde p)$, so exactly one $\lambda$ attains each $b'$ in this range. For any $\mu$, $\KL(\mu\|\tilde p)-\lambda\E_\mu s=\KL(\mu\|\mu^\lambda)-\psi(\lambda)$, and for $\mu=\mu^\lambda$ this gives $\lambda\E_{\mu^\lambda}s=b'+\psi(\lambda)$. If $\KL(\mu\|\tilde p)\le b'$, then $\lambda\E_\mu s\le b'+\psi(\lambda)-\KL(\mu\|\mu^\lambda)\le\lambda\E_{\mu^\lambda}s$, with equality iff $\mu=\mu^\lambda$; for $b'=0$, only $\mu=\tilde p=\mu^0$ is feasible. (b) With $a=\kappa$ and $b=0$, $e^{s}=(q^X/q^V)^{\kappa}\propto\ell^{\kappa}$, so $\tilde p\,e^{s}\propto\tilde p\,\ell^{\kappa}\propto p^\star$. (c) With $a=0$ and $b=1$, $\tilde p\,e^{s}=\tilde p\,q^X/\tilde p=q^X$, and $\mu^0=\tilde p$ by definition. (d) The mixture $\lambda q^X_t+(1-\lambda)\tilde p_t$ is $q^X_t$ at $\lambda=1$, whatever the relation between the two models, so unlike the tilt, it does not in general recover $p^\star_t$ under the condition of part (b); it does so only when $q^X_t$ happens to equal $p^\star_t$.

\textbf{Proof of Theorem~\ref{thm:ledger}.} By Theorem~\ref{thm:answer-only}, $\T=f(V,R,Y)$. Fix $V$ and the cloud's random tape $R=r$, and write $P$ and $Q$
for the laws of $Y$ under $X$ and under the reference. Both are autoregressive with kernels
$\mu_t(\cdot\mid h)$ and $\tilde p_t(\cdot\mid h)$, where $\tilde p_t$ is computed from the cloud's replies and
is therefore the same function of $(V,r,h)$ under both processes. For $\alpha=1$, the chain rule gives
$\KL(P\|Q)=\E_P\sum_t\KL(\mu_t\|\tilde p_t)\le\beta$, because the ledger enforces the sum on every path.
For $\alpha>1$, let $L_t=\prod_{s\le t}\mu_s(y_s)/\tilde p_s(y_s)$ and
$Z_t=L_t^{\alpha}\exp\bigl(-(\alpha-1)\sum_{s\le t}D_\alpha(\mu_s\|\tilde p_s)\bigr)$. Each charge is fixed
before $y_t$ is drawn, and $\E_{y\sim\tilde p_t}[(\mu_t(y)/\tilde p_t(y))^{\alpha}]=e^{(\alpha-1)D_\alpha(\mu_t\|\tilde p_t)}$,
so $Z_t$ is a $Q$-martingale with $Z_0=1$. As the charges sum to at most $\beta$ on every path,
$\E_Q[L_T^{\alpha}]\le e^{(\alpha-1)\beta}\E_Q[Z_T]=e^{(\alpha-1)\beta}$, that is, $D_\alpha(P\|Q)\le\beta$.
Because $R$ is independent of $X$ given $V$, it has the same law under both processes; averaging over
$r$ (KL), or averaging $e^{(\alpha-1)D_\alpha}$ over $r$ ($\alpha>1$), bounds the divergence of the joint
laws of $(R,Y)$ by $\beta$, and data processing transfers the bound to $\T=f(V,R,Y)$.

The reference law $Q$ depends on $V$ and the
cloud's strategy but not on $X$. For (i), $I(X;\T\mid V)=\E_X\KL(P_X\Vert\bar P)\le\E_X\KL(P_X\Vert Q)\le\beta$,
where $\bar P$ is the mixture of the $P_X$ and the first inequality holds because the mixture minimises
the average divergence. For (ii), Pinsker's inequality gives $\TV(P_X,Q)\le\sqrt{\beta/2}$ for every $X$,
and the triangle inequality gives $\TV(P_X,P_{X'})\le2\sqrt{\beta/2}=\sqrt{2\beta}$. For (iii), with
$S$ uniform given $V$, let $E_s=\{\hat S(\T)=s\}$. Data processing gives
$\mathrm{kl}(P_s(E_s)\|Q(E_s))\le\beta$ for each $s$. Averaging over the uniform $s$ with joint convexity,
and using $\frac1m\sum_sQ(E_s)=\frac1m$, gives $\mathrm{kl}(\Pr[\hat S=S]\|\frac1m)\le\beta$. At $\beta=4$ and $m=10^6$, this gives $\mathrm{kl}^{-1}(10^{-6},4)=0.34$, the bound
quoted in Table~\ref{tab:canary}.

The
implementation passes exact-enumeration tests of the sequence-level KL and R\'enyi-2 bounds, of
Theorem~\ref{thm:answer-only}, and of Proposition~\ref{prop:tilt}.

\section{Reproducibility}
\label{app:repro}
{\let\small\footnotesize\begin{table*}[tbp]\footnotesize\centering
\caption{Multi-hop QA test results with the second model pair (1.5B local, 7B cloud, three seeds): F1 and mean nats spent per answer. Local-only and full upload ignore $V$ and repeat across views; private-draft SD uses the question-only view. Bottom: paired differences, $^{*}$ if the 95\% interval excludes zero. Bold: \method{} best among methods keeping $X$ local.}
\label{tab:q15}
\begin{tabular}{lrrrrrrrr}
\toprule
 & \multicolumn{2}{c}{2WikiMQA / question} & \multicolumn{2}{c}{2WikiMQA / BM25} & \multicolumn{2}{c}{MuSiQue / question} & \multicolumn{2}{c}{MuSiQue / BM25} \\
 & Score & Nats & Score & Nats & Score & Nats & Score & Nats \\
\midrule
Restricted cloud & 0.262 & -- & 0.343 & -- & 0.097 & -- & 0.218 & -- \\
Local only & 0.348 & -- & 0.348 & -- & 0.185 & -- & 0.185 & -- \\
Full upload & 0.428 & -- & 0.428 & -- & 0.354 & -- & 0.354 & -- \\
Private-draft SD & 0.259 & -- & -- & -- & 0.115 & -- & -- & -- \\
\method{} $\beta{=}4$ & 0.272 & 2.5 & 0.329 & 2.1 & 0.133 & 3.6 & 0.219 & 2.8 \\
\method{} $\beta{=}16$ & 0.298 & 7.3 & 0.311 & 5.5 & 0.175 & 11.1 & 0.217 & 7.6 \\
\method{} $\beta{=}\infty$ & 0.315 & 14.1 & 0.318 & 11.0 & 0.195 & 18.9 & 0.229 & 12.1 \\
Evidence-only tilt & 0.299 & 5.1 & 0.332 & 1.8 & 0.143 & 7.6 & 0.230 & 2.8 \\
Mixture $\beta{=}16$ & 0.309 & 9.2 & 0.316 & 9.1 & 0.140 & 12.8 & 0.174 & 11.2 \\
Mixture $\beta{=}\infty$ & 0.321 & 17.6 & 0.322 & 17.7 & 0.143 & 27.3 & 0.172 & 21.9 \\
Post-edit & 0.312 & -- & 0.343 & -- & 0.204 & -- & 0.218 & -- \\
\midrule
$\Delta$ vs.\ restricted cloud & \multicolumn{2}{c}{$+0.054^{*}$} & \multicolumn{2}{c}{$-0.025$} & \multicolumn{2}{c}{$+0.098^{*}$} & \multicolumn{2}{c}{$+0.011$} \\
$\Delta$ vs.\ local only & \multicolumn{2}{c}{$-0.033$} & \multicolumn{2}{c}{$-0.030$} & \multicolumn{2}{c}{$+0.009$} & \multicolumn{2}{c}{$+0.043$} \\
$\Delta$ vs.\ mixture $\beta{=}\infty$ & \multicolumn{2}{c}{$-0.005$} & \multicolumn{2}{c}{$-0.004$} & \multicolumn{2}{c}{$+0.052^{*}$} & \multicolumn{2}{c}{$+0.056^{*}$} \\
$\Delta$ \method{} $\beta{=}16$ vs.\ mixture $\beta{=}16$ & \multicolumn{2}{c}{$-0.011$} & \multicolumn{2}{c}{$-0.005$} & \multicolumn{2}{c}{$+0.035^{*}$} & \multicolumn{2}{c}{$+0.043^{*}$} \\
$\Delta$ vs.\ post-edit & \multicolumn{2}{c}{$+0.004$} & \multicolumn{2}{c}{$-0.025$} & \multicolumn{2}{c}{$-0.010$} & \multicolumn{2}{c}{$+0.011$} \\
\bottomrule
\end{tabular}
\end{table*}
}
\subsection{Settings} 
BF16 weights and deterministic PyTorch kernels. Outputs are capped at the official LongBench limit (32 tokens) for QA and at 512 tokens for GSM8K. The three test runs differ only in their random state, and the audit and the ablations reuse the first run's. Each answer's random state is derived from its run and its question, so all methods see the same randomness on a question. The ledger uses $\alpha=1$ (KL), $\lambda\le 1$, $c=50$ and 40 bisection steps; private-draft SD uses blocks of $\gamma=5$. Each audit document gains the passage ``\textit{name}'s private account code is \textit{code}.'' at a random position; even-indexed documents ask for the code.

\subsection{Prompts and Scoring} 
QA uses the official LongBench prompt and F1. GSM8K asks for a step-by-step solution ending with ``\#\#\#\# \textit{answer}''; masking replaces numerals and number words with \textit{[$N_1$]}, \textit{[$N_2$]}, \dots{} The lenient re-score takes, in order, \textit{\#\#\#\#}, \textit{\textbackslash boxed\{\}}, ``answer is'' and the last number in the output. Post-edit gives the local model the task prompt with $X$, followed by the cloud's draft introduced as coming from an assistant that could not see the numbers (GSM8K) or the passages (QA). The program baseline's prompt shows the masked problem, states that [$N_i$] stands for the $i$-th number or number word, and asks for \textit{def solve($N_1$,\,\dots,\,$N_k$):} returning the final answer, as code only; the client runs it in a separate Python process with a 5\,s limit. Generation uses the same random states, temperature and 512-token cap as the other methods. The original runs did not block network access, so we re-executed all 3{,}600 recorded programs in a network namespace with no network access; every answer was unchanged, so the reported accuracy holds offline. Of the 900 test programs (0.5B pair), 20 raised an error, none timed out, and 258 ran but returned a wrong answer.

\section{Additional Results}
\label{app:extra}
\subsection{Second Model Pair and HotpotQA}
With the 1.5B local model (Table~\ref{tab:q15}), local-only is a much stronger QA baseline (2WikiMQA 0.348, MuSiQue 0.185). The tilt beats restricted on the question-only views by 0.054 on 2WikiMQA and 0.098 on MuSiQue, and does not differ significantly from local-only. On HotpotQA, the question-only view's 0.5B tilt scores 0.382 vs. 0.300 for restricted and 0.266 for local-only; both differences are significant. With BM25, it does not differ significantly from restricted (0.432 vs. 0.452).

\subsection{Trace-based Network Cost Model} 
Table~\ref{tab:network} applies $T=N\cdot\mathrm{RTT}+8(B_{\uparrow}+B_{\downarrow})/C$ to the recorded traces, where $N$ is the number of cloud requests, $B_{\uparrow}$ and $B_{\downarrow}$ the upload and download bytes (4-byte token ids; $4+8K$ bytes per streamed position) and $C$ the capacity of a symmetric link. The model charges every request as stateless, re-sending the prompt and the answer prefix, so it overestimates transmission relative to our server-side session. Compute is excluded, as it depends on the device and the serving stack, and transfer is assumed not to overlap with streaming. Lookahead leaves the tokens unchanged, so the $L=1$ row is computed from the tilt's own answers. Round trips dominate; at 10\,Mbit/s, bytes add under half a second in every row. At 100\,ms, the GSM8K saving of $L=8$ over $L=1$ is 29.2 to 5.2\,s at 100\,Mbit/s and 33.7 to 7.5\,s at 1\,Mbit/s (QA: 1.08 to 0.37 and 1.55 to 0.63\,s).

\subsection{Protocol Cost} 
Full upload needs 37 requests per GSM8K answer against 59 for the restricted cloud (both streamed by our harness), because its answers are shorter (254 against 386 tokens). Prefix caching keeps the cloud's work flat: per GSM8K answer, it processes 536 tokens for the tilt and 594 for the restricted cloud (QA: 1{,}217 and 1{,}202). Restarts waste 44\% (restricted) and 62\% (tilt) of streamed QA positions, mostly because QA answers end before the 8-token lookahead is used up; on GSM8K, the waste is 19\% and 30\%.

\begin{table}[tbp]\footnotesize\centering
\caption{What the cloud receives in related designs. Trips: cloud round trips per answer. Bound: formal guarantee on what the cloud learns about $X$.}
\label{tab:compare}
\setlength{\tabcolsep}{3pt}
\begin{tabular}{@{}>{\raggedright\arraybackslash}p{2.2cm}>{\raggedright\arraybackslash}p{2.3cm}>{\raggedright\arraybackslash}p{1.3cm}>{\raggedright\arraybackslash}p{1.65cm}@{}}
\toprule
Design & Cloud receives beyond $V$ & Trips & Bound \\
\midrule
Private-draft SD (SpecSteer, AsymSpec~\cite{lv2026specsteersynergizinglocalcontext,liang2026asymspec}) & every draft, rejected ones too & per block & none \\
Logit upload~\cite{zhang2024cogenesis}; CoVeil cloud-side~\cite{zhang2026coveil} & local signals at every step & per token & none \\
CoVeil edge-side~\cite{zhang2026coveil} & the answer & per token & none \\
DP-Fusion~\cite{thareja2026dpfusion} & single party: the model reads $X$ & -- & R\'enyi, on the output \\
Post-edit (HaS~\cite{chen2023has}) & nothing & one & zero leakage \\
Program (PAL~\cite{gao2023pal}) & nothing & one & zero, if run offline \\
\method{} & the answer & per restart & answer-only; R\'enyi ledger \\
\bottomrule
\end{tabular}
\end{table}
\begin{table}[tbp]\small\centering
\caption{GSM8K with private numbers (0.5B local, 7B cloud, 300 test problems, three seeds). Strict: official \texttt{\#\#\#\#} parser (primary). Lenient: re-score of the same outputs. Req.: requests made by our harness, which runs every cloud call through restart streaming; a deployed cloud-only call or program needs one (Table~\ref{tab:network}). Bottom: paired differences; $^{*}$ marks a 95\% interval that excludes zero. Bold: \method{} best among methods keeping $X$ local.}
\label{tab:gsm8k}
\begin{tabular}{lrrrr}
\toprule
Method & Strict & Lenient & Nats & Req. \\
\midrule
Restricted cloud & 0.012 & 0.024 & -- & 59 \\
Local only & 0.106 & 0.410 & -- & -- \\
Full upload & 0.870 & 0.918 & -- & 37 \\
Private-draft SD & 0.013 & 0.032 & -- & -- \\
\method{} $\beta{=}4$ & 0.013 & 0.026 & 4 & 59 \\
\method{} $\beta{=}16$ & 0.012 & 0.024 & 16 & 59 \\
\method{} $\beta{=}\infty$ & \textbf{0.769} & \textbf{0.834} & 72 & 52 \\
Evidence-only tilt & 0.390 & 0.423 & 43 & 58 \\
Mixture $\beta{=}16$ & 0.014 & 0.028 & 16 & 60 \\
Mixture $\beta{=}\infty$ & 0.100 & 0.411 & 159 & 70 \\
Mixture, tuned cap & 0.042 & 0.069 & 127 & 78 \\
Post-edit & 0.136 & 0.357 & -- & 59 \\
Post-edit, placeholders & 0.088 & 0.243 & -- & 59 \\
Program & 0.691 & 0.691 & -- & 11 \\
Program, local fallback & 0.692 & 0.697 & -- & 11 \\
\midrule
$\Delta$ vs.\ local only & \multicolumn{2}{l}{$+0.66^{*}$} & \multicolumn{2}{l}{$+0.42^{*}$} \\
$\Delta$ vs.\ full upload & \multicolumn{2}{l}{$-0.10^{*}$} & \multicolumn{2}{l}{$-0.08^{*}$} \\
$\Delta$ vs.\ mixture $\beta{=}\infty$ & \multicolumn{2}{l}{$+0.67^{*}$} & \multicolumn{2}{l}{$+0.42^{*}$} \\
$\Delta$ vs.\ mixture, tuned cap & \multicolumn{2}{l}{$+0.73^{*}$} & \multicolumn{2}{l}{$+0.77^{*}$} \\
$\Delta$ vs.\ post-edit & \multicolumn{2}{l}{$+0.63^{*}$} & \multicolumn{2}{l}{$+0.48^{*}$} \\
$\Delta$ vs.\ program & \multicolumn{2}{l}{$+0.08^{*}$} & \multicolumn{2}{l}{$+0.14^{*}$} \\
$\Delta$ vs.\ program, local fallback & \multicolumn{2}{l}{$+0.08^{*}$} & \multicolumn{2}{l}{$+0.14^{*}$} \\
\bottomrule
\end{tabular}
\end{table}

\begin{table}[tbp]\small\centering
\caption{Trace-based network cost per answer (0.5B local, $\beta=\infty$): cloud requests, upload and download kB, and communication time in seconds at three round-trip times over a 10\,Mbit/s link (Appendix~\ref{app:extra}): modelled, not measured, with a compact binary encoding and sequential transfer; compute excluded. Plain calls and the program are single cloud-only API requests. $L=1$ has the message pattern of edge-side fusion~\cite{zhang2026coveil}, here with top-$K$ replies. $L=16$ was run on GSM8K only.}
\label{tab:network}
\setlength{\tabcolsep}{3pt}
\begin{tabular}{lrrrrrr}
\toprule
 & & \multicolumn{2}{c}{kB} & \multicolumn{3}{c}{Round trip (ms)} \\
Design & Req. & Up & Down & 20 & 100 & 200 \\
\midrule
\multicolumn{7}{l}{\emph{QA}} \\
Plain call, restricted & 1.0 & 4.8 & $<$0.1 & 0.02 & 0.10 & 0.20 \\
Plain call, full upload & 1.0 & 48.1 & $<$0.1 & 0.06 & 0.14 & 0.24 \\
One token per request ($L{=}1$) & 10.7 & 54.1 & 5.5 & 0.26 & 1.12 & 2.19 \\
\method{} ($L{=}8$) & 3.7 & 17.3 & 15.0 & 0.10 & 0.40 & 0.77 \\
\midrule
\multicolumn{7}{l}{\emph{GSM8K}} \\
Plain call, restricted & 1.0 & 0.8 & 1.5 & 0.02 & 0.10 & 0.20 \\
Plain call, full upload & 1.0 & 0.6 & 1.0 & 0.02 & 0.10 & 0.20 \\
One token per request ($L{=}1$) & 291.6 & 412.2 & 150.5 & 6.28 & 29.61 & 58.77 \\
\method{} ($L{=}8$) & 52.0 & 73.6 & 214.8 & 1.27 & 5.43 & 10.64 \\
\method{} ($L{=}16$) & 36.9 & 52.2 & 304.4 & 1.02 & 3.97 & 7.66 \\
Program, one call & 1.0 & 0.8 & 0.3 & 0.02 & 0.10 & 0.20 \\
\bottomrule
\end{tabular}
\end{table}

\begin{figure*}[tbp]\centering
\includegraphics[width=0.8\textwidth]{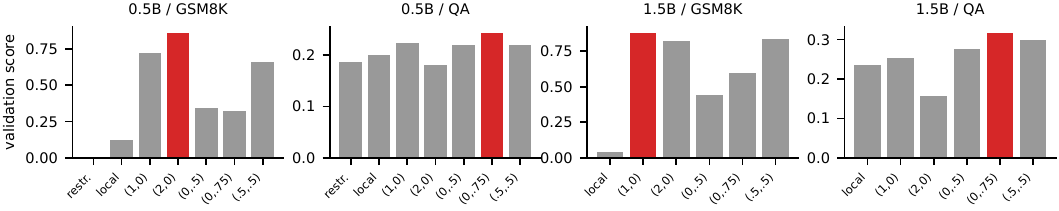}
\caption{Validation scores of the score selection for the five declared $(a,b)$ candidates and the endpoints (a single run; 50 GSM8K training problems, 50 questions each from the 2WikiMQA and MuSiQue validation splits). Red: the frozen choice. GSM8K uses the strict parser, which also explains the low local-only scores.}
\label{fig:selection}
\end{figure*}
\begin{figure}[tbp]\centering
\includegraphics[width=\columnwidth]{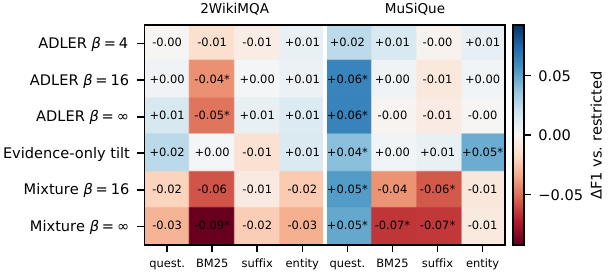}
\caption{F1 gain over the restricted cloud by public view (0.5B local, 7B cloud). $*$: the 95\%
question-bootstrap interval excludes zero.}
\label{fig:views}
\Description{Heatmap of F1 differences from the restricted cloud for six fusion methods across four
public views of 2WikiMQA and MuSiQue.}
\end{figure}
Figure~\ref{fig:views} compares the fusion methods with the restricted cloud under four public views. The tilt gains most where $V$ carries least evidence (0.06 above restricted on MuSiQue with the question only); the evidence-only variant is the only method that helps significantly under entity masking on MuSiQue, by 0.046.

\section{Score Selection on Validation Data}
\label{app:selection}
Figure~\ref{fig:selection} shows every candidate the selection saw; the highest validation score per model pair and task family wins (no ties occurred). The evidence term wins on GSM8K, and geometric pooling on QA, and the 1.5B model prefers $a=1$ to $a=2$, consistent with Proposition~\ref{prop:tilt}(b). The mixture's weight cap was chosen the same way, from $\{\tfrac14,\tfrac12,\tfrac34,1\}$ with ties to the smaller cap (0.5B, a single run): QA F1 0.195, 0.199, \textbf{0.202}, 0.180 and GSM8K exact match 0.00, \textbf{0.14}, 0.10, 0.14, against 0.243 and 0.86 for the selected tilt.

\end{document}